\documentclass{article}
\usepackage{fullpage}

\usepackage[utf8]{inputenc} %
\usepackage[style=alphabetic,natbib=true,maxcitenames=3, maxbibnames=10,backend=bibtex]{biblatex}
\usepackage[T1]{fontenc}    %
\usepackage[hypertexnames=false]{hyperref}
\usepackage{url}            %
\usepackage{booktabs}       %
\usepackage{amsfonts}       %
\usepackage{nicefrac}       %
\usepackage{microtype}      %
\usepackage{xcolor}         %
\usepackage{amsmath}
\usepackage{cleveref}
\usepackage{cases}
\usepackage{amsthm}
\usepackage{xspace}
\usepackage{graphicx}
\usepackage{comment}
\usepackage{enumitem}
\usepackage{tikz}
\usepackage{nicefrac}
\usepackage{thm-restate}
\usepackage{pgfplots}
\usepackage{subcaption}
\usepackage{nicefrac}
\usepackage{MnSymbol}
\usepackage[final]{showlabels}
\usepackage{mathtools}
\usepackage[framemethod=TikZ]{mdframed}
\mdfdefinestyle{proofofclaim}{%
    linecolor=black,
    topline=false,
    bottomline=false,
    rightline=false,
    leftline=true,
    }

\pgfplotsset{compat=1.18}
\usetikzlibrary{calc, positioning, arrows.meta}
\newcommand{\PURIFY}{\textup{\textsc{Purify}}\xspace}
\newcommand{\NOR}{\textup{\textsc{Nor}}\xspace}
\newcommand{\CFont}[1]{{\textup{{\textsf{#1}}}}\xspace}

\newcommand{\PC}{\CFont{PureCircuit}}
\newcommand{\EOTL}{\CFont{End-of-the-Line}}
\newcommand{\GDA}{\CFont{GDA-FixedPoint}}
\newcommand{\Pclass}{\CFont{P}}
\newcommand{\PPAD}{\CFont{PPAD}}
\newcommand{\CLS}{\CFont{CLS}}
\newcommand{\threeGDA}{\CFont{Degree-$3$-GDA}}
\newcommand{\twoGDA}{\CFont{Degree-$2$-GDA}}
\newcommand{\kGDA}{\CFont{Degree-$k$-GDA}}

\newcommand{\Reals}{\mathbb{R}}
\newcommand{\poly}{\textnormal{poly}}

\newcommand{\cG}{\mathcal{G}}

\newcommand{\cH}{\mathcal{H}}
\newcommand{\LINVI}{\CFont{LinVI}}
\newcommand{\TTPG}{\CFont{TwoTeam-ZeroSum-Polymatrix}}

\newcommand{\Naturals}{\mathbb{N}}
\newcommand{\G}{\textsf{G}}
\newcommand{\cD}{\mathcal{D}}
\newcommand{\KKT}{\textsf{KKT}}

\newtheorem{theorem}{Theorem}[section]
\newtheorem{lemma}{Lemma}

\theoremstyle{definition}
\newtheorem{definition}{Definition}
\newtheorem{problem}{Problem}
\newtheorem{remark}{Remark}

\newtheorem*{theorem*}{Theorem}

\Crefname{problem}{Problem}{Problems}
\Crefname{claim}{Claim}{Claim}

\definecolor[named]{linkcolor}{cmyk}{0.55,1,0,0.15}
\definecolor[named]{citecolor}{cmyk}{1,0.58,0,0.21}
\hypersetup{
    colorlinks=true,
    linkcolor=linkcolor,
    filecolor=linkcolor,      
    citecolor=citecolor,
    urlcolor=black,
}

\title{The Complexity of Min-Max Optimization for Quadratic Polynomials}

\author{
\begin{tabular}{cc}
& \\
{Martino Bernasconi}\thanks{Martino Bernasconi and Andrea Celli were supported by an ERC grant (Project 101165466 — PLA-STEER).} & {Matteo Castiglioni}\thanks{Matteo Castiglioni was supported by the FAIR (Future Artificial Intelligence Research) project PE0000013, funded by the NextGenerationEU program within the PNRRPE-AI scheme (M4C2, Investment 1.3, Line on Artificial Intelligence), and by the EU Horizon project ELIAS (European Lighthouse of AI for Sustainability, No. 101120237).}\\
\small{Bocconi University} & \small{Politecnico di Milano}\\
{\textcolor{black}{\small\texttt{martino.bernasconi@unibocconi.it}}} & %
{\textcolor{black}{\small\texttt{matteo.castiglioni@polimi.it}}}\\
& \\
{Andrea Celli}\footnotemark[1] & {Alexandros Hollender}\\
\small{Bocconi University} & \small{University of Oxford}\\
{\textcolor{black}{\small\texttt{andrea.celli2@unibocconi.it}}} & %
{\textcolor{black}{\small\texttt{alexandros.hollender@cs.ox.ac.uk}}}\\
& \\
\end{tabular}
}
\date{}

\begin{document}

\maketitle

\begin{abstract}
We prove that computing approximate stationary points of min-max optimization over the hypercube is \PPAD-hard for quadratic polynomials. This holds even when the polynomials are multilinear, each variable appears in at most three monomials, and the approximation factor is inverse polynomial. As a direct consequence, we obtain the first \PPAD-hardness results for two-team zero-sum polymatrix games.
\end{abstract}

\newpage
\section{Introduction}

We study the complexity of finding local solutions to constrained min-max optimization problems of the form \begin{equation}\label{eq:minmax}
     \min_{x \in X} \max_{y \in Y} f(x,y),
 \end{equation}
where $X,Y \subseteq \mathbb{R}^{d}$ are compact convex sets and $f : X \times Y \to \mathbb{R}$ is continuously differentiable with Lipschitz-continuous gradients. 
Recent work has shown that for general smooth nonconvex-nonconcave objectives, even computing first-order stationary points is computationally intractable \citep{bernasconi2026complexity,bernasconi2026min}. 
This is the natural extension of Karush–Kuhn–Tucker (KKT) points of minimization problems to min-max optimization (see also \citet*{daskalakis2021complexity} where the computational version of this question first appeared). Formally, an $\varepsilon$-approximate (constrained) stationary points is a pair $(x,y)\in X\times Y$ such that 
\begin{align}\label{eq:intro_kkt}
-\nabla_x f(x,y)^\top(x'-x)\le\varepsilon \quad \text{and}\quad \nabla_y f(x,y)^\top(y'-y)\le\varepsilon \quad \forall (x',y')\in X\times Y.
\end{align}
Equivalently, $x$ satisfies the approximate KKT conditions of the minimization problem $f(\cdot, y)$ and $y$ satisfies the approximate KKT conditions of the maximization problem $f(x, \cdot)$.
All known hardness results for min-max optimization for product constraints hold for very complicated functions $f$, which are either defined by a circuit \citep{bernasconi2026complexity}, or given by an oracle \citep{bernasconi2026min}. This leaves open whether the hardness is a consequence of such expressiveness, or whether it already appears for simple algebraic objectives.
On the positive side, this problem is efficiently solvable when $f$ is linear (or even convex-concave) in both variables \citep{rodomanov2023subgradient,facchinei2003finite}, and it remains reasonably tractable even in the more general case in which it is nonconvex in $x$ and concave in $y$ \citep{ostrovskii2021efficient}, namely it is solvable in time $\poly(1/\varepsilon, d)$. Hence, the simplest case in which the hardness of \Cref{eq:intro_kkt} could materialize is when $f$ is a degree-$2$ polynomial of the form
\[
f(x,y)=x^\top Q_1x+x^\top A y-y^\top Q_2y+b^\top x+c^\top y,
\]
and the first genuinely nonconvex-nonconcave quadratic regime is obtained when $Q_1$ and $Q_2$ are not positive semidefinite (since for $Q$ positive semidefinite, the form $x\mapsto x^\top Q x$ is convex). This leads to the central question addressed in this paper: 
\begin{quote}
\emph{Does the hardness of nonconvex-nonconcave min-max optimization already arise for degree-$2$ polynomials?}
\end{quote}
In this paper, we answer this question by showing that hardness persists even for multilinear degree-$2$ polynomials, thereby tightly identifying the natural algebraic boundary of tractability.

\paragraph{Team Games.} Quadratic min-max objectives also admit a natural game-theoretic interpretation, via their connection to two-team zero-sum polymatrix games. In a team game, players are partitioned into teams, and all members of the same team share a common utility, but each player controls only their own action, precluding joint deviations \citep{von1997team}. A very natural class of team games restricts the interactions between different teams to be zero-sum. Moreover, the polymatrix restriction imposes that each player's utility is separable, i.e., can be written as the sum of the utilities obtained from interacting with each other player. This class was introduced by \citet{bergman1998separable}, and \citet{cai2011minmax} proved that computing an equilibrium is \PPAD-hard with \emph{three} teams. However, they left the two-team zero-sum case open. 

The two-team case is particularly compelling because analogies from normal-form games suggest it might be tractable: while three-player zero-sum games are computationally intractable, the two-player version is solvable in polynomial time.
Two-team zero-sum games were recently studied in \citet*{hollender2025complexity,anagnostides2026complexity}, who showed \CLS-hardness for computing $\epsilon$-approximate equilibria when $\epsilon$ is exponentially small (i.e., ruling out algorithms running in time $\poly(\log(1/\varepsilon),d)$, unless $\CLS = \Pclass$). \CLS-completeness was also established in some special cases: \citet{hollender2025complexity} showed \CLS-completeness when one team has independent players, and \citet{anagnostides2026complexity} showed \CLS-completeness for $2$ vs.~$1$ games. These results, however, leave the full complexity picture open for two reasons. First, since Nash equilibria of finite games lie in \PPAD, which is believed to be a strict super-class of \CLS, \CLS-hardness does not settle the computational complexity of the problem. Second, these hardness results do not preclude the existence of an algorithm with running time $\poly(1/\varepsilon,d)$. In contrast, the \PPAD-hardness results in our paper apply to the $\varepsilon^{-1} = \poly(d)$ regime, and thus preclude the existence of a $\poly(1/\varepsilon,d)$ time algorithm, unless $\PPAD=\Pclass$.

\paragraph{Concurrent Work.} Similar results were obtained in a concurrent work by \citet*{anagnostides2026computational}.

\subsection*{Our Results and Techniques} 
\begin{figure}[t]
	\centering
	\scalebox{0.7}{ 
		\begin{tikzpicture}
		\tikzset{>={Latex}}
		\tikzstyle{cc}=[font=\normalsize, fill=blue!8,%
            rounded corners=4pt, inner sep=5pt,
            thick]
        \def\mergeone{2.3}       %
		\def\branchX{\mergeone + 2.3}      %
		\def\branchY{1.2}      %
		\def\mergeX{\branchX + 2.3}       %
		\def\threeGDAx{\mergeX + 2.3}    %
		\def\twoGDAx{\threeGDAx + 4.2}     %
		\def\ttpgX{\twoGDAx + 5.2}       %
        \def\bottom{-2.3cm}

		\node[cc] (eotl)     at (0, 0)            {\EOTL};
		\node[cc] (linvi)    at (\branchX,  \branchY) {\LINVI};
		\node[cc] (pc)       at (\branchX, -\branchY) {\PC};
		\node[cc] (threegda) at (\threeGDAx, 0)    {\threeGDA};
		\node[cc] (twogda)   at (\twoGDAx,   0)    {\twoGDA};
		\node[cc] (ttpg)     at (\ttpgX,     0)    {\TTPG};

		\coordinate (merge) at (\mergeX, 0);
        \coordinate (mergeone) at (\mergeone, 0);

        \draw[-, thick] (eotl.east) -- (mergeone);
		\draw[->, thick] (mergeone) -- (linvi.west);
		\draw[->, thick] (mergeone) -- (pc.west);

		\draw[-, thick] (linvi.east) -- (merge);
		\draw[-, thick] (pc.east)    -- (merge);
		\draw[->, thick] (merge)     -- (threegda);

		\draw[->, thick] (threegda) -- (twogda);
		\draw[->, thick] (twogda)   -- (ttpg);

        \node[font=\normalsize, fill=gray!0,
            rounded corners, inner sep=5pt] (or)     at (\mergeX, 0)            {$\lor$};

        \draw[{Bar[width=12pt]}-{Bar[width=12pt]}] ([yshift=\bottom]eotl.center) --  node[fill=white]{\Cref{sec:three}}([yshift=\bottom]threegda.center);
        \draw[-{Bar[width=12pt]}] ([yshift=\bottom]threegda.center) -- node[fill=white]{\Cref{sec:3to2}}([yshift=\bottom]twogda.center);
        \draw[-{Bar[width=12pt]}] ([yshift=\bottom]twogda.center) -- node[fill=white]{\Cref{sec:degree3interaction,sec:polymatrix,sec:2vs2}} ([yshift=\bottom]ttpg.center);
            
		\end{tikzpicture}
	}
	\caption{We start from an instance of \EOTL and reduce it to an instance of \LINVI and an instance of \PC. Then we use both instances to reduce them to an instance of \threeGDA. Next, we use the cubic-gadget to reduce the degree of the cubic terms to $2$ and thus obtain an instance of \twoGDA. Lastly, we use a multilinearization gadget to obtain a multilinear degree-$2$ instance, which is then converted to a \TTPG instance (first to one with many players and binary actions and then to a $2$ vs.~$2$ game where each player has many actions).}\label{fig:reductions}
\end{figure}
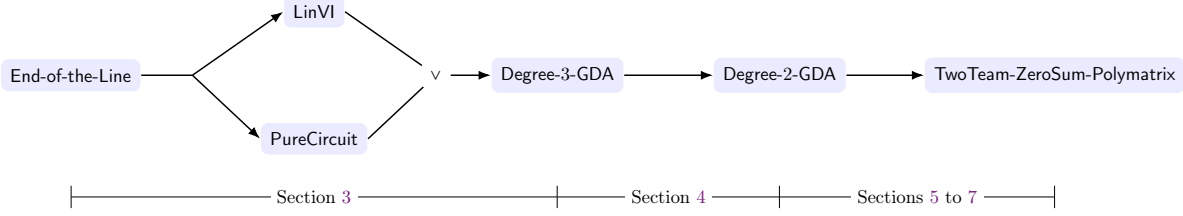

Our reduction follows several steps, summarized in \Cref{fig:reductions}, and can be roughly divided into three main parts. We now discuss the technical challenges and the main conceptual ideas of these three components.

\paragraph{Degree-$3$ Polynomials.} Our first result is a reduction showing that min-max optimization is \PPAD-hard even when the objective function is a degree-$3$ polynomial.

\begin{theorem*}[Informal version of \Cref{th:degree3minmax}]
    Finding a stationary point in a min-max optimization problem is \PPAD-hard even when the objective is a degree-$3$ polynomial.
\end{theorem*}

The reduction starts from \EOTL. As in earlier min-max reductions \citep{bernasconi2026complexity,bernasconi2026min}, we use two intermediate \PPAD-hard problems. The first one is \PC \citep{deligkas2022pure}, which is a Boolean circuit problem with self-loops, and the second one is \LINVI \citep{bernasconi2024role}, which is a simple linear variational inequality problem. The reduction combines them in the following way: either the constructed min-max instance faithfully transmits the \PC gates' signals, in which case we decode a \PC solution, or some gate fails to transmit its signal, in which case the failure exposes a solution to the \LINVI instance. We provide a brief overview of the relevant parts of the construction.
For each \PC node $q$, the construction introduces a block of variables 
\[
(x^q,y^q)=((x_i^q,y_i^q))_{i\in[n]}
\]
which are to be interpreted as being $n$ copies of $m$-dimensional variables $x_i^q,y_i^q$.
Then, the logic of the \PC instance (which we will refer to as the ``outer problem'') is used to define a signal $s_q\in[0,1]$ based on the values of the gate's input, which represents the Boolean value carried by that node. In particular, an intermediate signal is computed using a sigmoid-like function around $3m$ for the quantity $\|x^w-y^w\|^2$ for each input node $w$ of $q$. This intermediate signal can be interpreted as the approximate Boolean value carried by node $w$ (it is close to $0$ when the block distance is below $3m$, and close to $1$ when the block distance is above $3m$). These intermediate thresholded values are then fed into a second, gate-specific smooth threshold which implements the Boolean operation of the corresponding \PC gate and produces the output signal $s_q$.
Moreover, for each node, the reduction introduces an additional gadget 
\[
H_q(x,y)=\sum_{i\in[n]} \langle Dx_i^q+c_i,x_i^q-y_i^q\rangle
\]
based on the \LINVI instance (defined by the linear function $ \Reals^m\ni x\mapsto Dx+c\in\Reals^m$), which we call the ``inner problem''. The main idea of the \cite{bernasconi2026complexity} reduction was to use the inner problem, not to make every node solve the \LINVI problem directly, but rather to create a \emph{dichotomy}. At any stationary point of the min-max problem, we either have that all signals $s_q$ correctly implement the circuit logic, or any ``corrupted'' signal $s_q$ forces the gadget $H_q$ to reveal a solution to \LINVI at that node. 
The dichotomy is enforced by the regularization terms attached to the local \LINVI gadgets.  In the stationarity equations for a node $q$, the local \LINVI term is perturbed by derivative contributions from downstream gates (i.e., from gates in which $q$ appears as an input).  We denote the total such contribution by $\Delta_q$.  The regularizer introduces offsets $M_1,\ldots,M_n$, one for each copy of the gadget, which enter the stationarity conditions through the combination $M_i+\Delta_q$.  Thus, if $M_i$ is sufficiently close to the corruption $-\Delta_q$ for some copy $i$, the downstream perturbation is canceled in that copy, and the \LINVI structure can be recovered.
Indeed, the key technical challenge in \citet{bernasconi2026complexity,bernasconi2026min} was to show that the noise $\Delta_q$ corrupting the signal $s_q$ grows as $o(n)$. Then, since the offsets $M_i$ cover a range of size $O(n)$, one of them approximately cancels the perturbation.
The only part of this construction that was not encoded as a low-degree polynomial, in the original reduction of \citet{bernasconi2026complexity}, was the computation of the signals $s_q$, which were carried out by a small circuit (of depth $2$) of simple sigmoid-like gates.

Here, we use the natural idea of \emph{exponentially decreasing-weights}, used in the complexity of minimization problems \citep{krentel1990finding,babichenko2021settling,fearnley2022complexity} to substitute the sigmoid-like functions with a polynomial that carries out the relevant computations. As an example, we illustrate a simple \NOR gate with output $q$ and inputs $u,v$ (output is $1$ if both inputs are $0$, or $0$ if either input is $1$), which was implemented as:
\[
s_q(x,y)=\sigma^-_{1/2}(\sigma^+_{3m}(\|x^u-y^u\|^2)+\sigma^+_{3m}(\|x^v-y^v\|^2)),
\]
where $\sigma_{\alpha}^+$ is a sigmoid-like function which is close to $0$ for inputs below $\alpha$, close to $1$ for inputs above $\alpha$, and transitions smoothly between these two values near $\alpha$, 
and $\sigma_{\alpha}^-=1-\sigma_{\alpha}^+$. It is easy to check that, by interpreting the inputs according to thresholded values given by $\sigma^+_{3m}(\|x^u-y^u\|^2)$, $s_q(x,y)$ correctly implements the \NOR logic. Then, the signal $s_q$ was used in expressions such as $s_q(x,y)H_q(x,y)$ in the final construction.
Instead, we implement the signal using an exponentially decreasing-weight approach. Continuing with the previous example of a \NOR gate with inputs $u,v$, we proceed by introducing min variables $\bar x_u,\bar x_v,x_{u,v,q}$ and the polynomial
\begin{align*}
p(\bar x_u,\bar x_v, x_{u,v,q}, x^u,y^u,x^v,y^v)&=\bar x_u(3m-\|x^u-y^u\|^2)+\bar x_v(3m-\|x^v-y^v\|^2)+\\
&+\delta x_{u,v,q}(\bar x_u+\bar x_v-1/2)\\
&+\delta^2 x_{u,v,q} H_q(x,y).
\end{align*}
The idea is that, for sufficiently small $\delta$, higher-weight terms determine their variables before lower-weight terms can affect them.
Indeed, the gradient with respect to $\bar x_u$ is approximately $3m-\|x^u-y^u\|^2+\delta$, which implements $\bar x_u\approx \mathbb{I}(\|x^u-y^u\|^2\ge 3m)$ up to small error terms depending on $\delta$. Similarly, we fix $\bar x_v\approx \mathbb{I}(\|x^v-y^v\|^2\ge 3m)$. Once these terms are fixed, the next term forces $x_{u,v,q}\approx \mathbb{I}(\bar x_u+\bar x_v\le\tfrac12)$, which is exactly the \NOR gate's semantics.

Now this already implements a degree-$3$ function since the inner problem $H_q$ is naturally a degree-$2$ polynomial. 
This is where, differently from prior work, the choice of \LINVI as the inner problem becomes crucial. In the original reduction by \citet{bernasconi2026complexity} the use of \LINVI as the inner problem was mainly motivated by the fact that it yields a simpler gadget, and indeed in the subsequent work by \citet{bernasconi2026min} it was replaced by a Brouwer-type problem. For our reduction, the simplicity of the inner problem is fundamental rather than just convenient. In particular, the low degree of $H_q$ is what keeps the whole signal implementation at degree 3.  A more complicated inner gadget would no longer yield the desired degree-$3$ construction.

However, using these small polynomial gadgets to implement the signal $s_q$ poses challenges distinct from those in the previous reductions. The central challenge in \citet{bernasconi2026complexity} was to cover the corruption noise with the range of the guesses $M_i$. In our reduction, the focus shifts from establishing the existence of a copy on which $\Delta_q$ was correctly guessed to making sure that the signal correctly propagates in the circuit. In particular, the lower-order terms used to couple the computed signal to the \LINVI gadget must not interfere with the higher-priority terms that implement the Boolean gate. For instance, in our example, the derivative with respect to $x_{u,v,q}$ is $\partial_{x_{u,v,q}}p=\delta(\bar x_u+\bar x_v-1/2)+\delta^2H_q(x,y)$. So in order for $x_{u,v,q}$ to correctly implement the gate's logic, we need $\delta |H_q(x,y)|\ll 1$. Using the naive bound on $|H_q(x,y)|=O(n)$ (ignoring polynomial factors in $m$), we would require $\delta n=O(1)$. However, in order for \LINVI to correctly establish a dichotomy, we need $n$ to be large enough, so that the guesses $M_i$ cover the relevant corruptions. Therefore, a key new challenge of our reduction is to show a much sharper bound on the inner gadget, namely a bound on $|H_q(x,y)|$ that is independent of $n$. This ensures that the term $\delta^2 H_q(x,y)$ does not disturb the \NOR computation, while still allowing the \LINVI gadget to enforce the dichotomy.

\paragraph{Degree-$2$ Polynomials.}
Having proven hardness for degree-$3$ polynomials, it is natural to wonder whether the problem remains hard for lower degree polynomials. As mentioned earlier, for degree-$1$ polynomials the problem is easy, so the only remaining case is that of degree-$2$ polynomials, i.e., quadratic polynomials. The quadratic case differs from degree $3$ or higher in a fundamental way: for a quadratic polynomial, there is always a stationary min-max point that is rational. On the other hand, for degree-$3$ polynomials, it is easy to construct examples where \emph{all} solutions are irrational. The takeaway from this observation is that there is no simple reduction from the degree-$3$ case to the degree-$2$ case. Namely, any such reduction will not map exact solutions of the degree-$2$ instance back to exact solutions of the degree-$3$ instance.

For the minimization problem (as opposed to the min-max problem we consider here), \CLS-hardness for quadratic polynomials was shown by \citet*{FearnleyGHS25-quadratic-KKT}. Their reduction proceeds as follows. Starting from a \CLS-hard two-dimensional instance of the minimization problem with a general smooth objective function $f$, they first construct a piecewise-linear arithmetic circuit $\tilde{f}$ that approximates $f$ in the following sense: any approximate KKT point of $\tilde{f}$ must be close to an approximate KKT point of $f$. Then, they construct a gadget that can be used to implement any gate of a piecewise-linear arithmetic circuit. Importantly, this gadget is a quadratic polynomial, and when multiple copies of this gadget are combined to implement each gate of the circuit $\tilde{f}$, the resulting objective function is a quadratic polynomial $Q$. The crucial point in their reduction is that when these gadgets are combined, a \emph{backpropagation} phenomenon occurs. Namely, at any approximate solution, it is not only the case that the value of $Q$ is close to $\tilde{f}$, but additionally, the derivatives of $Q$ with respect to the input variables are also close to the corresponding derivatives of $\tilde{f}$. This ensures that any approximate KKT point of $Q$ yields an approximate KKT point of the original function $f$.

Can we use this approach in our setting? Unfortunately, their approach heavily relies on the fact that their aim is to show hardness for $\varepsilon$-KKT points, where $\varepsilon^{-1} = 2^{d}$. Indeed, the minimization problem is easy when $\varepsilon^{-1} = \poly(d)$, so $\varepsilon^{-1} = 2^{d}$ is the interesting regime in their setting. But for min-max optimization, our aim is to show hardness even when $\varepsilon^{-1} = \poly(d)$. Now, in order to use their approach to show hardness for $\varepsilon^{-1} = \poly(d)$, one would have to make the following modifications. In the first step of the reduction, one would have to show how to construct a piecewise-linear arithmetic circuit $\tilde{f}$ that approximates a \PPAD-hard \emph{high-dimensional} function $f$. This part was already very tedious in their setting for two-dimensional functions. It is not clear how to generalize this to high-dimensional functions, and even if it is possible, it is clear that it would be extremely technical and tedious. Although the second part of their reduction can be used more readily, an additional complication arises from using their gadget in our $\varepsilon^{-1} = \poly(d)$ regime: the circuit $\tilde{f}$ has to have constant depth, and this makes the task of the first step even harder.

Instead of embarking on this very complicated route, we adopt a significantly simpler approach. Starting from a \PPAD-hard degree-$3$ polynomial, we try to eliminate the degree-$3$ terms by replacing each of them by a degree-$2$ polynomial that behaves in a similar way. Using standard identities, it is not hard to see that terms of the form $x_1x_2x_3$ can be expressed as linear combinations of various cubic terms $x_1^3, (x_1+x_2)^3$, etc. As a result, the bottleneck of this approach is to find a way to replace a term of the form $x_1^3$ by a quadratic polynomial $P$. This polynomial $P$ would contain the variable $x_1$, as well as additional auxiliary variables $z_1, \dots, z_m$. Consider any KKT point $z^\star$ of $P(x_1,\cdot)$, i.e., when $x_1$ is considered fixed and thus the KKT conditions are only taken with respect to $z_1, \dots, z_m$. Then, we would like to have that
\[
\left| \partial_{x_1}P(x_1,z^\star) - 3x_1^2 \right|,
\]
is very small. In other words, the polynomial $P$ behaves like $x_1^3$ when we take derivatives with respect to $x_1$, or any other original variables $x_2, \dots, x_d, y_1, \dots, y_d$ (since those variables do not appear in $P$). This means that we can safely replace $x_1^3$ by the quadratic polynomial $P(x_1,z)$, since at any solution of our problem, all auxiliary variables will satisfy their KKT conditions, and thus the first-order derivatives with respect to the original variables will be correctly approximated.

It turns out that using ideas from \citep{FearnleyGHS25-quadratic-KKT}, and especially the backpropagation gadget, such a polynomial $P$ that behaves like $x_1^3$ can indeed be constructed. In fact, this idea can even be used to provide an alternative, somewhat simpler proof of their result: essentially, one can follow their general recipe, but apply it to the one-dimensional function $x_1^3$, instead of a two-dimensional unknown smooth function $f$. This considerably simplifies various aspects of their proof, although some additional work is required to ensure the resulting quadratic polynomial has derivative close to $3x_1^2$, as opposed to just being close to zero when $3x_1^2$ is close to zero (which is what their existing reduction enforces).

Thankfully, since we are working in the $\varepsilon^{-1} = \poly(d)$ regime, we can afford a significantly larger error in how closely the derivative of $P$ approximates $3x_1^2$. Thus, it turns out that the desired polynomial $P$ can be constructed in a much more direct manner. In essence, it suffices to consider a sufficiently good piecewise-constant approximation of the function $3x_1^2$, and then to construct a polynomial $P$ whose derivative closely follows the piecewise-constant approximation. This can be achieved in a very direct and elegant way.
Thus, we obtain the following result.
\begin{theorem*}[Informal version of \Cref{th:degree-2-non-multilinear}]
    Finding a stationary point in a min-max optimization problem is \PPAD-hard even when the objective is a degree-$2$ polynomial.
\end{theorem*}

\paragraph{Two-Teams Zero-Sum Polymatrix Games.} After proving hardness for the degree-$2$ version of the problem, in \Cref{sec:degree3interaction} we begin considering game-theoretic applications, in particular to polymatrix team games.
First, we convert a hard instance of a degree-$2$ polynomial into a multilinear polynomial of degree-$2$ (i.e., there are no monomials of degree $2$, in only one variable), and which is \emph{sparse}, which means that each variable appears only in a few monomials.
\begin{theorem*}[Informal version of \Cref{th:degree2}]
    Finding a stationary point in a min-max optimization problem is \PPAD-hard even when the objective is a multilinear degree-$2$ polynomial, and each variable appears only in three monomials.
\end{theorem*}
To do so, we extend the idea of the \emph{copy-gadget} of \citet{hollender2025complexity} that was originally used to multilinearize a degree-$2$ instance. In particular, the original copy gadget uses the min-max structure to ``copy'' the value of a variable $x_k$ into a copy $x_k'$, by using an extra variable $\tilde y_k$. This is helpful in multilinearizing polynomials, since we can then replace a term of the form $x_k^2$ with a term of the form $x_kx_k'$, which is multilinear. The important property is not only that $x_k\approx x_k'$ at solutions, but also that the partial derivative with respect to $x_k$ is preserved at solutions.
We extend this idea to multiple copies: if $x_k$ appears in $D$ monomials, we make $D$ copies $\{x_k^{(1)},\ldots,x_k^{(D)}\}$ of it, by introducing only at most $D$ new max variables $\tilde y_k^{(j)}$, and we replace each appearance of $x_k$ with one of its copies. The challenge now is that in the new polynomial, the derivatives of all other copies should be propagated back to the derivative of the ``original'' copy $x_{k}^{(1)}$.

This result almost directly proves that two-team zero-sum games are \PPAD-hard and, in particular, implies the following theorem(s).
\begin{theorem*}[Informal version of \Cref{th:manyplayerpolymatrix} and \Cref{th:2vs2}]
    Finding an equilibrium of a two-team polymatrix game is \PPAD-hard even when there are polynomially many players with binary actions, or when there are only $2$ players per team, each with polynomially many actions.
\end{theorem*}
First, in \Cref{sec:polymatrix}, we show how to reduce a multilinear polynomial of degree $2$ to a polymatrix game while preserving sparsity.
Indeed, existing reductions from degree-$2$ to polymatrix (such as the one of \citet{hollender2025complexity}) do not preserve the sparsity, and we have to make simple modifications to those reductions in order to do so. This implies that polymatrix games with polynomially many players and binary actions, with only zero-sum or coordination edges (i.e., a two-team zero-sum polymatrix game) are \PPAD-hard.

Lastly, in \Cref{sec:2vs2}, we prove that the problem remains hard even if we limit the setting to having only $2$ players per team (but with polynomially many actions each). 
To do so, we use the observation that, in the construction of \Cref{sec:degree3interaction}, we can group all the $x$ (and $y$, resp.) variables into two groups such that all variables in one group never appear in a monomial together. Then we ``assign'' all the variables of one group to a single bit-player, and we let each group player first choose a player, then the action that player should take. With a standard ``hide-and-seek'' gadget game played on the side (see \citet{althofer1994sparse, rubinstein2017settling}), we can ensure that the players mix (approximately) uniformly over the choice of the bit-player to simulate, thus making sure that the equilibrium of the $2$ vs.~$2$ game leads back to an equilibrium of the original binary-action polymatrix game.

\section{Preliminaries}

In this paper, we study fixed points of Gradient Descent-Ascent (GDA) dynamics, or equivalently stationary points of min-max optimization. This solution concept is the natural generalization of KKT points in minimization problems to the more general setting of min-max optimization. We begin by defining approximate KKT points for minimization problems.

\begin{definition}
    Given a continuously differentiable function $g:[0,1]^d\to\Reals$ and $\varepsilon>0$, we define 
    \[
    \KKT_\varepsilon(g)=\{x\in[0,1]^d:-\partial_{x_i}g(x)(x_i'-x_i)\le\varepsilon\\,\forall i\in [d],\forall x_i'\in[0,1]\}
    \]
    to be the set of $\varepsilon$-KKT points of the minimization problem $\min_{x \in [0,1]^d} g(x)$.
\end{definition}
Then, we define \GDA as follows.

\begin{problem}[\GDA]
    Given $\varepsilon>0$ and two circuits implementing a $\poly(d)$-Lipschitz and $\poly(d)$-smooth function $f:[0,1]^d\times[0,1]^d\to[-1,1]$  and its gradient $\nabla f:[0,1]^{d}\times[0,1]^d\to[-1,1]^{2d}$, find a solution $x,y\in[0,1]^d$ such that
    \[
    x\in \KKT_\varepsilon(f(\cdot,y))\quad\text{and}\quad y\in \KKT_\varepsilon(-f(x,\cdot)).
    \]
\end{problem}

As mentioned in the introduction, \GDA was shown to be \PPAD-hard by \citet{bernasconi2026complexity}. In this paper, we are interested in the special case in which $f$ is a low-degree polynomial.

\begin{problem}[\kGDA]
    Given $\varepsilon>0$ and a degree-$k$ polynomial $f$ in $2d$ variables (represented by its coefficients, each in $[-1,1]$), find a solution to \GDA.
\end{problem}

We assume that polynomials are represented in their canonical form, where each monomial appears at most once.
Finally, we consider two-team zero-sum polymatrix games, which are fundamentally related to \kGDA with $k=2$. %

\begin{problem}[\TTPG]
    Given $\varepsilon>0$, a set of $\mathcal{N}$ players, partitioned in two teams $X\sqcup Y=\mathcal{N}$, and matrices $A^{i,j}\in[-1,1]^{d\times d}$, for all ${i,j\in \mathcal{N}}$ satisfying for all $i,i'\in X$ and $j,j'\in Y$,
    \[
    A^{i,i'} = (A^{i',i})^\top,\quad A^{j,j'} = (A^{j',j})^\top\quad\text{and}\quad A^{i,j} = -(A^{j,i})^\top,
    \]
    find $x=(x_1,\ldots, x_{|\mathcal{N}|})\in \Delta_d^{|\mathcal{N}|}$ such that
    \[
    U_i(x_i,x_{-i}) \ge U_i(x'_{i},x_{-i})-\varepsilon\quad\forall i \in \mathcal{N}, x_i'\in\Delta_d
    \]
    where $U_i(x_i,x_{-i})=\sum_{j\in \mathcal{N}}x_i^\top A^{i,j} x_j$.
\end{problem}

We will reduce from the \PPAD-complete problem \EOTL. %

\begin{problem}[\EOTL]
Given two circuits $S,P:\{0,1\}^N\to\{0,1\}^N$ (called, successor and predecessor circuit, respectively), such that $S(0)\neq P(0)=0$, find a node $v\in\{0,1\}^N$ such that $P(S(v))\neq v$ or $S(P(v))\neq v\neq 0$.
\end{problem}

We will also make use of the following two \PPAD-complete problems.

\begin{problem}[\PC\citep{deligkas2022pure}]
An instance of \PC is given by a vertex set $V = [\kappa]$ and two sets of gates $\cG_{\NOR}$ and $\cG_{\PURIFY}$. Each gate is of the form $(u, v, w)$ where $u, v, w \in V$ are distinct nodes with the following interpretation:
\begin{itemize}
\item If $(u,v,w) \in \cG_{\NOR}$, then $u$ and $v$ are the inputs of the gate, and $w$ is its output.
\item If $(u,v,w) \in \cG_{\PURIFY}$, then $u$ is the input of the gate, and $v$ and $w$ are its outputs.
\end{itemize}

Each node is the output of exactly one gate.
A solution to an instance of \PC is an assignment $b : V \rightarrow \{0, 1, \bot\}$ that satisfies all the gates,
i.e., for each gate $(u, v, w) \in \cG$ we have:
\begin{itemize}
    \item if $(u, v, w) \in \cG_{\NOR}$, then $b$ satisfies
    \begin{align*}
        &b(u) = b(v) = 0 \implies b(w) = 1\\
        &(b(u) = 1) \text{  or  } (b(v) = 1) \implies b(w) = 0
    \end{align*}
    
    \item if $(u, v, w) \in \cG_{\PURIFY}$, then $b$ satisfies
    \begin{align*}
        &\{b(v), b(w)\} \cap \{0,1\} \neq \emptyset\\
        &b(u) \in \{0,1\} \implies b(v) = b(w) = b(u)
    \end{align*}
\end{itemize}
\end{problem}

\begin{problem}[\LINVI\citep{bernasconi2024role}]
    Given a matrix $D\in[-1,1]^{m\times m}$, a vector $c \in  [-1,1]^m$, and an approximation parameter $\rho>0$, find a point $z\in [0,1]^m$ such that:%
   \[(Dz+c)_j\cdot (z_j'-z_j)\ge -\rho \quad \forall j \in [m], z'_j \in [0,1]. \]
\end{problem}

\citet{bernasconi2024role} show that \LINVI is \PPAD-complete for a constant $\rho>0$.

\section{\threeGDA}\label{sec:three}

As in prior hardness results for min-max optimization \cite{bernasconi2026complexity,bernasconi2026min}, our reduction relies on two \PPAD-hard problems: an inner problem based on linear variational inequalities (\LINVI), and an outer problem based on discrete fixed-point computations (\PC).

\subsection{Construction}
We reduce from \EOTL. Given an instance of \EOTL, we construct an instance of the \PPAD-complete problem \LINVI, which is defined by a matrix $D\in[-1,1]^{m\times m}$, a vector $c\in[-1,1]^m$, and an approximation parameter $\rho>0$. Moreover, we construct an instance of the \PPAD-Complete problem \PC, which is defined by a collection of gates $\cG_\NOR$ and $\cG_\PURIFY$ over $\kappa=|V|$ nodes.

We let $n\in\Naturals$ and $\varepsilon,\delta>0$ be defined in the following, and build the following instance of \threeGDA with approximation $\varepsilon$:

\begin{itemize}
    \item For each $v\in V$, we introduce $2nm+1$ variables $(x^v_{i,j})_{i\in [n],j\in[m]}$, $(y^v_{i,j})_{i\in [n],j\in[m]}$ and $\bar x_v$. For convenience, we write $x_i^v:=(x_{i,j}^v)_{j\in[m]}$ and $x^v:=(x_{i,j}^v)_{i\in[n],j\in[m]}$, and similarly for $y_i^v$ and $y^v$. 
    \item For each $(u,v,w)\in\cG_{\NOR}$, we introduce the variable $x_{u,v,w}$.
    \item For each $(u,v,w)\in\cG_{\PURIFY}$, we introduce the variables $x_{u,v,w}'$ and $x_{u,v,w}''$. 
    \item For each node $v\in V$, define
    \[
    H_v(x,y)=\sum_{i\in[n]} \langle D x_i^v+c,x_i^v-y_i^v\rangle.
    \]
    \item For each $q\in V$, define
    \[
    s_q(x)=
    \begin{cases}x_{u,v,q}&\text{if}\quad(u,v,q)\in \cG_{\NOR}\\
    x'_{u,q,w}&\text{if}\quad(u,q,w)\in \cG_{\PURIFY}\\
    x''_{u,v,q}&\text{if}\quad(u,v,q)\in \cG_{\PURIFY}
    \end{cases},
    \]
    which is well defined since every node is the output of exactly one gate.
    \item For each $i\in[n]$, let $M_i=i/n$, so that $M_1,\ldots, M_n$  is an arithmetic progression with common difference $1/n$ starting at $1/n$.

    \item The objective function $f$ is defined as:
    \begin{align}
    f(x,y)&=\delta^2\sum_{q\in V} s_q(x)H_q(x,y)+\tag{signal}\\
    &+\delta\sum_{\substack{(u,v,w)\in\cG_{\NOR}}} x_{u,v,w}\big(\bar x_u+\bar x_v-\tfrac12\big)\tag{\NOR gate variables}\\
    &+\delta\sum_{(u,v,w)\in\cG_{\PURIFY}}( x_{u,v,w}'(\tfrac14-\bar x_u)+ x_{u,v,w}''(\tfrac34-\bar x_u))\tag{\PURIFY gate variables}\\
    &+\sum_{u \in V} \bar x_u(3m-\|x^u-y^u\|^2)\tag{node variables}\\
    &+\sum_{v\in V}\sum_{i\in[n]} M_i\|x_i^v-y_i^v\|^2\tag{guesses}
    \end{align}

    \item The parameters are set to $\cH:=\left(\frac{m\kappa}{\rho}\right)^{10}$, $\delta:=\frac{1}{100m\kappa \cH}$, $n:=\lceil\frac{400m^3}{\rho^2\delta^2}\rceil$, and $\varepsilon:=\frac{\rho\delta^2}{100m^2n^3\cH^2}$.
\end{itemize}

\subsection{Derivatives}
It will be convenient to write explicitly the partial derivatives of $f$, which will be used throughout the analysis.

\paragraph{Gate variables.} For the gate variables  $x_{u,v,w}$, $x'_{u,v,w}$ and $x''_{u,v,w}$, direct computations give the following.
For $(u,v,w)\in \cG_\NOR$,
\begin{align}\label{eq:derivuvwone}
\partial_{x_{u,v,w}}f(x,y)=
\delta^2 H_w(x,y)+\delta (\bar x_u+\bar x_v-\tfrac12),
\end{align}
For $(u,v,w)\in \cG_{\PURIFY}$,
\begin{align}\label{eq:derivuvwtwo}
\partial_{x'_{u,v,w}}f(x,y)=\delta^2H_v(x,y)+\delta(\tfrac14-\bar x_u),
\end{align}
and
\begin{align}\label{eq:derivuvwthree}
\partial_{x''_{u,v,w}}f(x,y)=\delta^2H_w(x,y)+\delta(\tfrac34-\bar x_u).
\end{align}

\paragraph{Threshold variables.} Each node $q\in V$ has associated a variable $\bar x_q$. Define, for each $q\in V$, \[N_q^\NOR(x)=\sum_{(q,v,w)\in\cG_\NOR}x_{q,v,w}+\sum_{(u,q,w)\in\cG_\NOR}x_{u,q,w},\quad\text{and}\quad N_q^\PURIFY(x)=\sum_{(q,v,w)\in\cG_\PURIFY}(x'_{q,v,w}+x''_{q,v,w}).\]
Then:
\begin{align}\label{eq:deriv_node}
\partial_{\bar x_q}f(x,y) = (3m-\|x^q-y^q\|^2)+\delta(N_q^\NOR(x)-N_q^\PURIFY(x)).
\end{align}

\paragraph{Main variables.} Since $H_q$ only depends on $x^q$ and $y^q$, for each $q\in V$ and $(i,j)\in[n]\times[m]$, we get:
\[
\partial_{x_{i,j}^q} f(x,y)=\delta^2 s_q(x) \partial_{x_{i,j}^q} H_q(x,y)+2(M_i-\bar x_q)(x^q_{i,j}-y^q_{i,j}),
\]
and 
\[
\partial_{y_{i,j}^q} f(x,y)=\delta^2 s_q(x) \partial_{y_{i,j}^q} H_q(x,y)-2(M_i-\bar x_q)(x^q_{i,j}-y^q_{i,j}).
\]

\subsection{Tight bound on $|H_q(x,y)|$}\label{sec:smallH}
In this section, we show that $|H_q(x)|$ is bounded by $\cH=\poly(m\kappa/\rho)$ (and in particular that it is independent of $n$) at any solution of \threeGDA. The following estimate is the key point at which the construction differs from the earlier smooth-threshold reductions.
A trivial bound gives $|H_q|=O(n)$, since $H_q$ is a sum over $n$ copies.

We will rely on two lemmas, both of which appeared in a similar form in \citet{bernasconi2026complexity,bernasconi2026min}.
The first lemma captures a simple observation: for copies $i\in[n]$ and nodes $q$ where the term $(M_i-\bar x_q)$ is large, the two players have aligned objectives and thus will play similar strategies. The proof is almost identical to \citet[Lemma 7.1]{bernasconi2026complexity} (which, however, cannot be used black-box) and it is deferred to the Appendix \ref{sec:missing3}.
\begin{restatable}{lemma}{xminusy}\label{lem:xminusy}
    For any solution $(x,y)$ of \threeGDA and any $(i,j,q)\in[n]\times[m]\times V$ with $M_i\neq \bar x_q$, it holds:
    \[
    |x_{i,j}^q-y_{i,j}^q|\le \frac{2\delta^2 s_q(x)m}{|M_i-\bar x_q|}+\sqrt{\frac\varepsilon{|M_i-\bar x_q|}}.
    \]
\end{restatable}

The second lemma is a simple upperbound of the sum of the reciprocals of a shifted arithmetic series. For every $x\in[0,1]$, let $B(x)=\{i\in[n]:|M_i-x|<1/n\}$. Then:

\begin{restatable}{lemma}{harmonic}\label{lem:harmonic}
    For any solution $(x,y)$ of \threeGDA, it holds:
    \[
    \sum_{i\in[n]\setminus B(x)}\frac1{|M_i-x|}\le4n\log n.
    \]
\end{restatable}

With these two lemmas, we are ready to state and prove the main result of this section.

\begin{lemma}\label{lem:boundH}
    For any solution $(x,y)$ of \threeGDA, it holds 
    \[
    |H_q(x,y)|\le 16\delta^2m^3n\log(n)+2m^2n^{3/2}\sqrt{\varepsilon}+4m^2\le \cH.
    \]
\end{lemma}

\begin{proof}
We first bound $|H_q(x,y)|$ in terms of $\|x^q_{i}-y^q_{i}\|_1$. We have
\begin{align*}
    |H_q(x,y)|=\Bigg|\sum_{i\in [n]} \langle D x^q_i+c,x^q_{i}-y^q_{i}\rangle\Bigg| \le \sum_{i\in[n]}\|D x^q_i+c\|_\infty\|x^q_{i}-y^q_{i}\|_1\le 2m\sum_{i\in[n]} \|x^q_{i}-y^q_{i}\|_1,
\end{align*}

If $i\notin B(\bar x_q)$ then by \Cref{lem:xminusy} we get:
\begin{align*}
    \|y^q_{i}-x^q_{i}\|_1&=\sum_{j\in[m]}|x_{i,j}^q-y_{i,j}^q|\\
    &\le \frac{2\delta^2 m^2}{|M_i-\bar x_q|}+m\sqrt{\frac{\varepsilon}{|M_i-\bar x_q|}}\\
    &\le \frac{2\delta^2 m^2}{|M_i-\bar x_q|}+m\sqrt{n\varepsilon}.
\end{align*}

Then, we can use \Cref{lem:harmonic} and get
\begin{align*}
    \sum_{i\notin B(\bar x_q)} \|x_i^q-y_i^q\|_1&\le 2\delta^2m^2\sum_{i\notin B(\bar x_q)}\frac{1}{|M_i-\bar x_q|}+mn^{3/2}\sqrt{\varepsilon}\\
    &\le 8\delta^2m^2n\log(n)+mn^{3/2}\sqrt{\varepsilon}.
\end{align*}

If $i\in B(\bar x_q)$ then $\|x_i^q-y_i^q\|_1\le m$, and $|B(\bar x_q)|\le 2$ so this contributes at most $2m$.

Combining the two cases, we get:
\[
\|x^q-y^q\|_1\le 8\delta^2m^2n\log(n)+mn^{3/2}\sqrt{\varepsilon}+2m,
\]
and therefore
\[
|H_q(x,y)|\le 16\delta^2m^3n\log(n)+2m^2n^{3/2}\sqrt{\varepsilon}+4m^2\le \cH,
\]
where the last inequality comes from our choice of $n,\delta$ and $\varepsilon$.%
\end{proof}

\subsection{Dichotomy Lemma for \threeGDA}\label{sec:dichotomy}

In this section, we will prove the main technical tool underlying our reduction. As in previous min-max hardness results \citep{bernasconi2026complexity,bernasconi2026min}, a solution to \threeGDA is only guaranteed to recover either a solution to \PC or a solution to the inner problem \LINVI.  The key insight is that the signal $s_q$ is encoded in the distance between $x^q$ and $y^q$: a small signal forces the two to be close, while a large signal forces them apart, unless the \LINVI is solved.

We first handle the small-signal case.
It is easy to prove, by \Cref{lem:xminusy}, that when the signal from $q$ is small, then $\|x^q-y^q\|^2$ is itself small.

\begin{lemma}\label{lem:smallsignal}
   For any solution $(x,y)$ of \threeGDA, if $s_q(x)\le \varepsilon/(\delta^2\cH)$ then $\|x^q-y^q\|^2\le 3m-10\delta\kappa$.%
\end{lemma}
\begin{proof}
    Consider any $i\not\in B(\bar x_q)$, which by definition satisfies $|M_i-\bar x_q|\ge 1/n$. By \Cref{lem:xminusy} we have 
    \begin{align*}
        |x_{i,j}^q-y_{i,j}^q|&\le \frac{2\delta^2 s_q(x)m}{|M_i-\bar x_q|}+\sqrt{\frac{\varepsilon}{|M_i-\bar x_q|}}\\
        &\le \frac{2 \varepsilon m }{\cH|M_i-\bar x_q|}+\sqrt{\frac{\varepsilon}{|M_i-\bar x_q|}}\\
        &\le \frac{2\varepsilon nm}{\cH}+\sqrt{\varepsilon n}\\
        &\le \frac1{mn}.
    \end{align*}
    Thus, for $i\notin B(\bar x_q)$ we have $\|x_i^q-y_i^q\|_2^2\le 1/n$.
   Consider now an $i\in B(\bar x_q)$. then trivially $\|x_i^q-y_i^q\|_2^2\le m$ and $|B(\bar x_q)|\le 2$. 
   Combining the two upper bounds we get:
    \[
    \|x^q-y^q\|_2^2\le 2m+1\le 3m-10\delta\kappa,
    \]
    where the last inequality holds since $\delta\kappa=\frac{1}{100m\cH}\le \frac1m$.
\end{proof}

For the large-signal case, a symmetric statement would require $\|x^q-y^q\|^2$ to be large whenever $s_q$ is large. However, this cannot be guaranteed unconditionally.  Indeed, we can only prove that this holds only if there is no index $i$ such that $y_i^q$ solves the corresponding \LINVI problem. 
This is also where we need $n$ to be large enough. Indeed, when $s_q$ is close to $1$, the \LINVI term is active on all copies $i$ for which
\[
|M_i-\bar x_q|<\delta^2.
\]
If no such copy yields an approximate \LINVI solution, then each of these copies must contribute a nontrivial amount to
$\|x^q-y^q\|^2$. Therefore, we need the window of radius $\delta^2$ to contain many grid points, i.e., we need $n\delta^2$ to be large. This is the reason for choosing $n$ of order at least $1/\delta^2$.
Formally:
\begin{lemma}\label{lem:largesignal}
For any solution $(x,y)$ of \threeGDA, if $s_q(x)\ge 1-\varepsilon/(\delta^2\cH)$, and no $i\in[n]$ exists such that $y_i^q$ is a $\rho$-approximate solution to \LINVI, then $\|x^q-y^q\|^2_2\ge 3m+10\delta\kappa$.
\end{lemma}
\begin{proof}
    By assumption, $y_i^q$ is not a $\rho$-approximate solution to $\LINVI$ for all $i\in[n]$.
    Consider indexes $i$ with $|M_i-\bar x_q|<\delta^2$. Notice that since $n\delta^2\ge 2$ there are at least $n\delta^2/2$ such indexes. We claim that $\|x_i^q-y_i^q\|_1\ge\rho/10$ for such $i$. 
    
    Suppose by contradiction that $\|x_i^q-y_i^q\|_1<\rho/10$.
    The optimality conditions of the variables $y_{i,j}^q$, $j\in[m]$, reads as:
    \[
    \Big(-\delta^2 s_q(x)(Dx_i^q+c)_j -2(M_i-\bar x_q)(x^q_{i,j}-y^q_{i,j})\Big)(z-y_{i,j}^q)\le\varepsilon\quad\forall z\in[0,1].
    \]
    Rearranging, we  get:
    \begin{align*}
        -(Dx_i^q+c)_j(z-y_{i,j}^q)&\le\frac{\varepsilon+2(M_i-\bar x_q)(x_{i,j}^q-y_{i,j}^q)(z-y_{i,j}^q)}{\delta^2 s_q(x)}\\
        &\le \frac{2\varepsilon+4|M_i-\bar x_q|\cdot|x_{i,j}^q-y_{i,j}^q|\cdot|z-y_{i,j}^q|}{\delta^2 }\tag{$s_q(x)\ge 1/2$}\\
        &\le \frac{2\varepsilon}{\delta^2}+4\frac{|M_i-\bar x_q|}{\delta^2}|x_{i,j}^q-y_{i,j}^q|\\
        &\le \frac{2\varepsilon}{\delta^2}+4|x_{i,j}^q-y_{i,j}^q|\tag{$|M_i-\bar x_q|<\delta^2$}\\
        &\le \frac{2\varepsilon}{\delta^2}+\frac25\rho\le \frac12\rho\\
    \end{align*}
    Moreover, we also have:
    \[|D(y_i^q-x_i^q)_j|\le |\sum_{k\in[m]}D_{jk}(y_{i,k}^q-x_{i,k}^q)|\le\|x_i^q-y_i^q\|_1\le \frac\rho{10}.\]
    Combining these two inequalities, we get that for all $j\in[m]$ and $z\in[0,1]$:
    \begin{align*}
    (Dy_i^q+c)_j(z-y_{i,j}^q) &= D(y_i^q-x_i^q)_j(z-y_{i,j}^q)+(Dx_i^q+c)_j(z-y_{i,j}^q)\\
    &\ge -|D(y_i^q-x_i^q)_j|+(Dx_i^q+c)_j(z-y_{i,j}^q)\\
    &\ge -\rho.
    \end{align*}
    This shows that $y_i^q$ is a solution of the \LINVI instance, and thus $\|x_i^q-y_i^q\|_1<\rho/10$ is a contradiction.
    Now, recalling that there are at least $n\delta^2/2$ indexes $i\in[n]$ such that  $\|x_i^q-y_i^q\|_1\ge \rho/10$, we get:
    \[
    \|x^q-y^q\|^2\ge \sum_{i\in[n]:|M_i-\bar x_q|<\delta^2}\|x_i^q-y_i^q\|^2\ge \frac{1}{m}\sum_{i\in[n]:|M_i-\bar x_q|<\delta^2}\|x_i^q-y_i^q\|^2_1\ge \frac n{2m}\delta^2\frac{\rho^2}{100}\ge m^2\ge 3m+10\delta\kappa,
    \]
    concluding the proof of the lemma.
\end{proof}

\subsection{The \PC Gates are Evaluated Correctly}\label{sec:gatesarecorrect}

In this section we are going to prove that the node variables $\bar x_q$ and the gate variables (that are $x_{u,v,w},x_{u,v,w}'$ and $x_{u,v,w}''$) correctly implement the \PC constraints.  We begin with a general lemma relating the sign of a partial derivative to the value of the corresponding min variable at any solution.
\begin{lemma}\label{lem:deriv_pm}
    Let $z$ be any $\min$ variable (that is a variable assigned to the $x$ player). For any solution $(x,y)$ of \threeGDA, we have
    \begin{itemize}
        \item If $\partial_z f(x,y)\ge \gamma$ then $z\le \varepsilon/\gamma$,
        \item if $\partial_z f(x,y)\le -\gamma$ then $z\ge 1-\varepsilon/\gamma$.
    \end{itemize}
\end{lemma}
\begin{proof}
    By the optimality condition, $-\partial_z f(x,y)(z'-z)\le\varepsilon$ for all $z'\in[0,1]$ and thus also for $z'=0$. Thus, in the first case, we obtain $z\le\varepsilon/\partial_z f(x,y)\le \varepsilon/\gamma$.
    In the second case, evaluating the condition for $z'=1$, we get $\gamma(1-z)\le-\partial_zf(x,y)(1-z)\le\varepsilon$, from which follows that $z\ge 1-\varepsilon/\gamma$.
\end{proof}

Then, we show that the node and the gate variables implement the \PC constraints.

\paragraph{Node Variables.}

The node variables $(\bar x_q)_{q\in V}$ value depends on whether $\|x^q-y^q\|^2$ is greater or smaller than $3m$.
It is close to zero when $\|x^q-y^q\|^2$ is small, and close to one if $\|x^q-y^q\|^2$ is large. Formally:

\begin{lemma}\label{lem:nodevars}
    For any solution $(x,y)$ of $\threeGDA$, we have
    \begin{itemize}
    \item if $\|x^q-y^q\|^2\ge 3m+10\delta\kappa$ then $\bar x_q\ge 1-\varepsilon/(\delta\kappa)$,
    \item if $\|x^q-y^q\|^2\le 3m-10\delta\kappa$ then $\bar x_q\le \varepsilon/(\delta\kappa)$.
    \end{itemize}
\end{lemma}

\begin{proof}
First observe that $|N_q^\NOR(x)|\le 2\kappa$ and $|N_q^\PURIFY(x)|\le 2\kappa$. 
Then, we consider the two cases separately.
If $\|x^q-y^q\|^2\ge 3m+10\delta\kappa$, by \Cref{eq:deriv_node} we get 
\[
\partial_{\bar x_q}f(x,y)\le 4\delta\kappa-10\delta\kappa=-6\delta\kappa,
\]
and thus, by \Cref{lem:deriv_pm}, $\bar x_q\ge 1-\varepsilon/(\delta\kappa)$.

If $\|x^q-y^q\|^2\le 3m-10\delta\kappa$, then, by \Cref{eq:deriv_node}, we get
\[
\partial_{\bar x_q}f(x,y)\ge \delta\kappa,
\]
and, by \Cref{lem:deriv_pm}, $\bar x_q\le \varepsilon/(\delta\kappa)$.
\end{proof}

\paragraph{Gate Variables.}
Now we are going to show that the gate variables $(x_{u,v,w},x_{u,v,w}',x_{u,v,w}'')$ encode the output of the corresponding gate. The proof of this lemma is a simple application of \Cref{lem:deriv_pm}. This is the crucial step for which we require that $\delta\cH\ll 1$, and where the $\delta$-hierarchy replaces the smooth threshold circuits of the earlier reduction. For example, in a \NOR gate, the derivative of the output variable has the form
\[
\delta(\bar x_u+\bar x_v-\tfrac12)+\delta^2H_q.
\]
The first term is the intended gate signal, while the second term is the bounded perturbation from the inner gadget. Since $|H_q|\le\mathcal H$, a margin of order $10\delta\mathcal H$ in $\bar x_u+\bar x_v-\tfrac12$ is enough to determine the sign of the derivative.

\begin{lemma}\label{lem:gates}
Consider any solution $(x,y)$ of \threeGDA:%
    \begin{itemize}
        \item If $(u,v,w)\in\cG_\NOR$:
        \begin{itemize}
            \item If $\bar x_u+\bar x_v\ge \frac12+10\delta \cH\implies s_w(x)=x_{u,v,w}\le\varepsilon/(\delta^2\cH)$
            \item If $\bar x_u+\bar x_v\le \frac12-10\delta \cH\implies s_w(x)=x_{u,v,w}\ge1-\varepsilon/(\delta^2\cH)$
        \end{itemize}
        \item If $(u,v,w)\in\cG_\PURIFY$:
        \begin{itemize}
            \item If $\bar x_u\le\tfrac14-10\delta\cH \implies s_v(x)=x'_{u,v,w}\le\varepsilon/(\delta^2\cH)$
            \item If $\bar x_u\ge\tfrac14+10\delta\cH \implies s_v(x)=x'_{u,v,w}\ge1-\varepsilon/(\delta^2\cH)$
            \item If $\bar x_u\le\tfrac34-10\delta\cH \implies s_w(x)=x''_{u,v,w}\le\varepsilon/(\delta^2\cH)$
            \item If $\bar x_u\ge\tfrac34+10\delta\cH \implies s_w(x)=x''_{u,v,w}\ge1-\varepsilon/(\delta^2\cH)$
        \end{itemize}
    \end{itemize}
\end{lemma}
\begin{proof}
    We apply \Cref{lem:deriv_pm} to the derivative of each variable, using the tight bounds in \Cref{lem:boundH} to control on $|H_q(x,y)|$ terms. We consider two cases:
    \begin{itemize}
        \item $(u,v,w)\in\cG_\NOR$.
            Suppose that $\bar x_u+\bar x_v\ge \frac12+10\delta \cH$. Then, by \Cref{eq:derivuvwone}, we have that $\partial_{x_{u,v,w}}f(x,y)=\delta^2H_w(x,y)+\delta(\bar x_u+\bar x_v-\tfrac12)\ge \delta^2 H_w(x,y)+10\delta^2\cH \ge \delta^2\cH $ and thus, thanks to \Cref{lem:boundH}, we get $x_{u,v,w}\le\varepsilon/(\delta^2\cH)$.
            Suppose that $\bar x_u+\bar x_v\le \frac12-10\delta \cH$. Then $\partial_{x_u,v,w}f(x,y)\le -\delta^2\cH$ and $x_{u,v,w}\ge1-\varepsilon/(\delta^2\cH)$.
    
        \item $(u,v,w)\in\cG_\PURIFY$.
    First consider the variable $x'_{u,v,w}$.
             Suppose that $\bar x_u\le\tfrac14-10\delta\cH$. Then \Cref{eq:derivuvwtwo} implies that $\partial_{x_{u,v,w}'}f(x,y)\ge\delta^2\cH$ and by \Cref{lem:deriv_pm}, that $x_{u,v,w}'\le\varepsilon/(\delta^2\cH)$. Suppose that $\bar x_u\ge\tfrac14+10\delta\cH$. Then, $\partial_{x_{u,v,w}'}f(x,y)\le-\delta^2\cH$ and thus $x_{u,v,w}'\ge1-\varepsilon/(\delta^2\cH)$.
             Consider now the variable $x''_{u,v,w}$.
             Similarly, using \Cref{eq:derivuvwthree}, we can infer that if $\bar x_u\le\tfrac34-10\delta\cH$  then $x_{u,v,w}''\le\varepsilon/(\delta^2\cH)$ and that if $\bar x_u\ge\tfrac34+10\delta\cH$, then $x_{u,v,w}''\ge1-\varepsilon/(\delta^2\cH)$.
    \end{itemize}
    This concludes the proof.
\end{proof}

\subsection{\threeGDA is \PPAD-hard}

In this section, we combine the results of \Cref{sec:dichotomy} and \Cref{sec:gatesarecorrect} to show that \threeGDA is \PPAD-hard.
To do so, given a solution $
(x,y)$ of \threeGDA, we define the \PC assignment $b:V\to\{0,1,\bot\}$:
\begin{align}\label{eq:encoding}b(v):=
    \begin{cases}
        1&\text{if}\quad \|x^v-y^v\|^2\ge 3m+10\delta\kappa\\
        0&\text{if}\quad \|x^v-y^v\|^2\le 3m-10\delta\kappa\\
        \bot&\text{otherwise}
    \end{cases}.
\end{align}

Then, given a solution to \threeGDA, we show that either $b$ satisfies all the gate constraints or that we are able to recover a solution to \LINVI. This implies the \PPAD-harndess of \threeGDA.

\begin{theorem}\label{th:degree3minmax}
\threeGDA is \PPAD hard even for $\varepsilon=\frac{1}{\poly(d)}$.
\end{theorem}

\begin{proof}
Define $b:V\to\{0,1,\bot\}$ as in \Cref{eq:encoding}. To prove the theorem it is sufficient to show that either $b$ is a valid assignment to \PC, or there exists $i\in[n],q\in V$ such that $y_i^q\in[0,1]^m$ is a solution to \LINVI.
Then, we can easily normalize all the coefficients to $[-1,1]$ by decreasing $\varepsilon$ by a factor polynomial in $d$.

Now, notice that if there exists $i\in[n],q\in V$ such that $y_i^q\in[0,1]^m$ is a solution to \LINVI, then the proof is complete. So now we can assume that there is no such tuple $(i,q)$ and show that $b$ is a valid assignment to \PC.
We start from the \NOR gates.
Consider any $w\in V$ such that $(u,v,w)\in\cG_\NOR$. We need to guarantee that if $b(u)=b(v)=0$, then $b(w)=1$, and that if $b(u)=1$ or $b(v)=1$ then $b(w)=0$.
        \begin{enumerate}
            \item $b(u)=b(v)=0$. Then $\|x^v-y^v\|^2,\|x^u-y^u\|^2\le 3m-10\delta\kappa$. Therefore, from \Cref{lem:nodevars}, we get 
            $\bar x_u\le \varepsilon/(\delta\kappa)$, $\bar x_v\le \varepsilon/(\delta\kappa)$. Now we can apply \Cref{lem:gates} to show that since $\bar x_u+\bar x_v\le 2\varepsilon/(\delta\kappa)\le \tfrac12-10\delta\cH$ it holds  $s_w(x)=x_{u,v,w}\ge 1-\varepsilon/(\delta^2\cH)$. Now, since no $y^q_i$ is a solution to \LINVI, we can invoke \Cref{lem:largesignal} to conclude that $\|x^w-y^w\|^2\ge 3m+10\delta\kappa$. Hence, $b(w)=1$ by \Cref{eq:encoding}.
            
            \item  $b(u)=1$ or $b(v)=1$. We assume that $b(u)=1$ since the other case is symmetric. Then $\|x^u-y^u\|^2\ge 3m+10\delta\kappa$. Thus, by \Cref{lem:nodevars}, we have $\bar x_u\ge 1-\varepsilon/(\delta\kappa)$ and $\bar x_u+\bar x_v\ge 1/2+10\delta\cH$. Moreover, from \Cref{lem:gates}, we get $s_w(x)=x_{u,v,w}\le\varepsilon/(\delta^2\cH)$ and from \Cref{lem:smallsignal} that $\|x^w-y^w\|^2\le 3m-10\delta\kappa$. Hence, $b(w)=0$ by \Cref{eq:encoding}.
        \end{enumerate}
        
        Similarly, consider any $v,w\in V$ such that $(u,v,w)\in\cG_\PURIFY$. We need to guarantee that if $b(u)\in\{0,1\}$, then $b(u)=b(v)=b(w)$, and that if $b(u)=\bot$ then $b(v)\in\{0,1\}$ or $b(w)\in\{0,1\}$.
        \begin{enumerate}
            \item  $b(u)=0$. Then, $\|x^u-y^u\|^2\le 3m-10\delta\kappa$. Thus, by \Cref{lem:nodevars} we have $\bar x_u\le \varepsilon/(\delta\kappa)$ and by \Cref{lem:gates} that $s_v(x)=x'_{u,v,w}\le \varepsilon/(\delta^2\cH)$ and $s_w(x)=x''_{u,v,w}\le \varepsilon/(\delta^2\cH)$. This, thanks to \Cref{lem:smallsignal}, implies that $\|x^v-y^v\|^2,\|x^w-y^w\|^2\le 3m-10\delta\kappa$ and thus by \Cref{eq:encoding} we get $b(v)=b(w)=0$, as desired.
            \item $b(u)=1$. Similarly, we have $\|x^u-y^u\|^2\ge 3m+10\delta\kappa$. By \Cref{lem:nodevars}, we then have that $\bar x_u\ge 1-\varepsilon/(\delta\kappa)$ and by \Cref{lem:gates} that $s_v(x)=x'_{u,v,w}\ge 1-\varepsilon/(\delta^2\cH)$ and $s_w(x)=x''_{u,v,w}\ge 1-\varepsilon/(\delta^2\cH)$. Now, since no $y^q_i$ is a solution to \LINVI, we can invoke  \Cref{lem:largesignal}, which implies that both $\|x^v-y^v\|^2$ and $\|x^w-y^w\|^2$ are greater or equal to $3m+10\delta\kappa$. Thus by \Cref{eq:encoding} we get $b(v)=b(w)=1$, as desired.
            
            \item $b(u)=\bot$. Then, $\|x^u-y^u\|^2\in[3m-10\delta\kappa, 3m+10\delta\kappa]$. Thus, we have no guarantees on $\bar x_u$, and we only know that $\bar x_u\in[0,1]$. We now consider two cases. If $\bar x_u\le \tfrac34-10\delta\cH$, then $s_w(x)=x_{u,v,w}''\le\varepsilon/(\delta^2\cH)$ by \Cref{lem:gates}, and \Cref{lem:smallsignal} guarantees that $\|x^w-y^w\|^2\le3m-10\delta\kappa$. Thus, $b(w)=0$ by \Cref{eq:encoding}.
            If $\bar x_u\ge \tfrac14+10\delta\cH$, then $s_v(x)=x_{u,v,w}'\ge1-\varepsilon/(\delta^2\cH)$ by \Cref{lem:gates}. Since no $y^q_i$ is a solution to \LINVI, \Cref{lem:largesignal} guarantees that $\|x^v-y^v\|^2\ge3m+10\delta\kappa$ and thus $b(v)=1$ by \Cref{eq:encoding}. This proves that for all $\bar x_u$, either we have $b(v)=1$ or $b(w)=0$.
        \end{enumerate}
        This concludes the proof.
\end{proof}

We would like to remark further on the following subtlety in our construction. The choice of $n$ is constrained in two opposite ways. The large-signal dichotomy of \Cref{lem:largesignal} requires $n$ to be large. On the other hand, the gate analysis of \Cref{lem:gates} requires $\delta |H_q|\ll1$. Since a trivial estimate gives
$|H_q|=O(n)$, these two requirements would be incompatible. The refined bound on $H_q$ of \Cref{sec:smallH} is what makes the parameter choice possible.

\section{From \threeGDA to \twoGDA} \label{sec:3to2}

In this section, we provide a reduction from \threeGDA to \twoGDA. Our primary tool is \Cref{lm:x3}, stated below, which allows us to replace terms of the form $\theta^3$ by a quadratic polynomial $P$ that has derivative $\partial_\theta P$ close to $3\theta^2$. As alluded to in the introduction, this can then be used as a primitive to show that we can replace any degree-$3$ monomial by a corresponding degree-$2$ polynomial.

We note that, using ideas from prior work \citep{FearnleyGHS25-quadratic-KKT}, one can prove a different version of \Cref{lm:x3}, where the error is much smaller, namely $O(2^{-m})$, but at the cost of requiring $\varepsilon^{-1} = 2^{\poly(m)}$. This version can then be used to provide a somewhat simpler proof of the main result of \citet{FearnleyGHS25-quadratic-KKT}, starting from the \CLS-hardness of minimizing degree-$5$ polynomials proved by \citet{babichenko2021settling}. In the context of min-max optimization however, we are able to show hardness for the $\varepsilon^{-1} = \poly(d)$ regime and this allows us to work with the lemma as stated below, which has inverse polynomial error. This has the important advantage that the lemma is significantly easier to establish for this error regime.

\begin{lemma}[Cubic Gadget]\label{lm:x3}
Given any $m \in \mathbb{N}$, we can construct in $\poly(m)$-time a degree-$2$ polynomial $P(\theta,z_1, \dots, z_{m})$ in $m+1$ variables with all coefficients bounded by $1$ in absolute value, such that for any value $\theta \in [0,1]$ and any $z^*\in \KKT_\varepsilon(-P(\theta,\cdot))$, where $\varepsilon\le O(m^{-4})$,
we have
\[
\left| \frac{\partial  }{\partial\theta}P(\theta,z^*) - 3\theta^2 \right| \leq O\Big(\frac1m\Big).
\]
\end{lemma}
\begin{proof}

Our idea is to discretize the derivative $3\theta^2$ as follows. For each $i \in \{0,\ldots,m\}$, let $t_i=\frac{i}{m}$ and $a_i=\frac{6i}{m^2}$. Then, we define the function  
\[ \G(\theta)= \sum_{i\in [m]} a_{i} \mathbb{I}[\theta \ge t_i],\]
which approximates the derivative $3\theta^2$.
Indeed, given a variable $\theta \in [t_k,t_{k+1})$, i.e., $\theta \in [\frac{k}{m},\frac{k+1}{m})$, we get
\begin{align} \label{eq:G}
\G(\theta)= \sum_{i =1,\ldots,k} a_i= \frac{3k(k+1)}{m^2}=3 \left(\frac{k}{m}\right)^2 + \frac{3k}{m^2}. 
\end{align}

Now, we construct the following quadratic polynomial, whose derivative with respect to $\theta$ will approximate the function $\G$:
\[P(\theta,z)= \sum_{i \in [m]} a_i (\theta-t_i) z_i. \]

The idea is that $z_i\simeq 0$ if $\theta \ll t_i$ and  $z_i\simeq 1$ if $\theta \gg t_i$. 
Let $z^* \in \KKT_\varepsilon(-P(\theta,\cdot))$.
Consider a $\theta$, and $t_i$ such that $ \theta\in [t_i-\frac{1}{2m}, t_i+\frac{1}{2m}]$.
For each variable $z_j$, with $j<i$, we get that, by the KKT conditions of $z_j$ with respect to $1$: %
\[ a_j (\theta-t_j) (1-z^\star_j)\le \varepsilon \]
and hence, since $a_j(\theta-t_j)\ge a_j/(2m)$:%
\[z_j^\star \ge 1-  \frac{2\varepsilon m}{a_j} \ge 1-  2\varepsilon m^3. \]
Similarly, if $j>i$, then
\[z_j^\star \le  \frac{2\varepsilon m}{a_j} \le  2\varepsilon m^3. \]

Overall, we can conclude that
\begin{align*}
\left|\partial_{\theta}P(\theta,z^*)- 3\theta^2\right| &= \left|\sum_{j \in [m]} a_j z^\star_j - 3\theta^2\right| \\
&= \left| \sum_{j < i} a_j + \sum_{j < i} a_j\left(z_j^\star - 1 \right) + a_{i} z_i^\star + \sum_{j > i} a_j z_j^\star - 3\theta^2\right| \\
& \le \left| \sum_{j < i} a_j -3\theta^2\right| + O\left(m^3 \varepsilon + \frac{1}{m}\right)\\
& \le\left| 3 \left(\frac{i-1}{m}\right)^2 + \frac{3(i-1)}{m^2} -3 \left(\frac{i}{m}\right)^2 \right| + 3\left|\theta^2-\left(\frac{i}{m}\right)^2\right|+O\left(m^3 \varepsilon + \frac{1}{m}\right)\\
& \le O\left(m^3 \varepsilon + \frac{1}{m}\right),
\end{align*}
where in the second-to-last inequality we use \Cref{eq:G}.
Taking $\varepsilon= O(m^{-4})$, we get the statement.

\end{proof}

\Cref{lm:x3} allows us to replace cubic terms of the form $x_1^3$ by a degree-2 polynomial. More generally, the lemma allows us to replace $\theta^3$ by a degree-2 polynomial, as long as $\theta$ is some expression that always lies in $[0,1]$. For example, letting $\theta = (x_1+x_2+x_3)/3$, \Cref{lm:x3} allows us to replace $(x_1+x_2+x_3)^3$ by a degree-2 polynomial.
Hence, we are left to show that any degree-3 monomial can be expressed using only such cubic terms and that negative coefficients can also be handled.

The first issue is easily addressed by using the following polarization trick, which allows us to express terms of the form $abc$ solely as a linear combination of terms of the form $(\alpha a+\beta b+\gamma c)^3$.

\begin{lemma}[$3$-linear polarization identity \citep{thomas2014polarization}]\label{lem:polarization}
    Take any symmetric three-linear map $T:[0,1]^3\to\Reals$ and define the diagonalized map $T'(x)=T(x,x,x):[0,1]\to \Reals$. Then
    \[
    T(a,b,c)=\frac{1}{6}\left[27T'\left(\frac{a+b+c}{3}\right)-8T'\left(\frac{a+b}{2}\right)-8T'\left(\frac{b+c}{2}\right)-8T'\left(\frac{a+c}{2}\right)+T'(a)+T'(b)+T'(c)\right].
    \]
\end{lemma}

Now we have to account for the negative and positive coefficients in these terms, since our cubic-gadget only works with positive coefficients. By using simple algebraic identities, we can reduce \threeGDA to the following specialized form. We defer the proof of \Cref{lem:degree3props} to Appendix \ref{app:3to2}.

\begin{restatable}{lemma}{lemmaspecial}\label{lem:degree3props}
\threeGDA is \PPAD-hard even when $f(x,y)=q(x,y)+\sum_{k=1}^T\alpha_k L_k(x,y)^3$, where $\alpha_k\ge 0$, $T=\poly(d)$, $L_k$ is affine and depends on at most $3$ variables, $L_k(x,y)\in[0,1]$, $\varepsilon^{-1}=\poly(d)$ and $q$ is a degree-$2$ polynomial. Moreover, all $\alpha_k$ and all the coefficients of $q$ are in $[-1,1]$.
\end{restatable}

Finally, we can prove that   \twoGDA is \PPAD-hard by combining the cubic terms approximation in \Cref{lm:x3} with the symmetric polynomial in \Cref{lem:degree3props}.

\begin{theorem}\label{th:degree-2-non-multilinear}
    \twoGDA is \PPAD-hard even for $\varepsilon=\frac{1}{\poly(d)}$.
\end{theorem}
\begin{proof}
    Thanks to \Cref{lem:degree3props}, \threeGDA is \PPAD-hard even when restricted to degree-$3$ polynomials $f:[0,1]^d\times[0,1]^d\to[-1,1]$ which satisfy:
    \[
    f(x,y)=q(x,y)+\sum_{k=1}^T\alpha_k L_k(x,y)^3,
    \]
    and considering an approximation error $\varepsilon'=\frac{1}{\poly(d)}$.

    For each term $k\in[T]$, we use \Cref{lm:x3} to define the polynomial $P(L_k(x,y),\xi_k^1,\ldots,\xi_k^{m})$ which includes the extra max variables $\xi_k^1,\ldots,\xi_k^{m}\in[0,1]$, where $m$ will be chosen in the following.
    Then, overall, we get a polynomial 
    \[
    f'(x,y)=q(x,y)+\sum_{k=1}^T\alpha_k P(L_k(x,y),\xi_k^1,\ldots,\xi_k^{m}).
    \]

    This is a degree-$2$ polynomial and thus an instance of \twoGDA. 

    We are left to show that it is possible to recover a solution to $f$ from an approximate solution to $f'$.
    Let $\varepsilon$ be the approximation error of the \twoGDA instance, which will be set further below.
    Consider a solution $(x,y,\xi)$ to $f'$. Notice that all variables $ \xi_k^j$ are approximately satisfying the KKT condition, allowing us to exploit \Cref{lm:x3}.
    
    Let $u$ be any variable of the original instance, so $u\neq \xi_k^j$ for all $k\in[T],j\in[m]$. To keep the notation compact we denote with $a \pm b$ the interval $[a-b,a+b]$.
    Then 
    \begin{align*}
        \partial_{u} f'(x,y)&= \partial_uq(x,y)+\sum_{k=1}^T\alpha_k \partial_{u}P(L_k(x,y),\xi_k^1,\ldots,\xi_k^{m})\\
        &=\partial_uq(x,y)+\sum_{k=1}^T\alpha_k \partial_{L_k(x,y)}P(L_k(x,y),\xi_k^1,\ldots,\xi_k^{m})\partial_uL_k(x,y)\\
        &=\partial_uq(x,y)+\sum_{k=1}^T\alpha_k (3L_k(x,y)^2\pm O\left(m^{-1}\right))\partial_uL_k(x,y)\\
        &=\partial_uq(x,y)+3\sum_{k=1}^T\alpha_k L_k(x,y)^2\partial_uL(x,y)\pm O\left(m^{-1}\right)\sum_{k\in [T]}\alpha_k\partial_uL_k(x,y).
    \end{align*}
    Now observe that
    \begin{align*}
        \partial_uf(x,y)&=\partial_uq(x,y)+\sum_{k=1}^T\alpha_k \partial_u(L_k(x,y)^3)\\
        &=\partial_uq(x,y)+3\sum_{k=1}^T\alpha_k L_k(x,y)^2\partial_u L_k(x,y),
    \end{align*}
    and thus
    \[
    |\partial_u f(x,y)-\partial_u f'(x,y)|\le O\left(m^{-1}\right)\sum_{k\in[T]}\alpha_k\partial_u L_k(x,y).
    \]
    Hence, we get that for each minimization variable $x_i$ and each $x_i'\in[0,1]$
    \[
    -\partial_{x_i}f(x)(x_i'-x_i)\le -\partial_{x_i}f'(x)(x_i'-x_i) + O\left(m^{-1}\right)\sum_{k\in[T]}\alpha_k\partial_u L_k(x,y) \le  \varepsilon + O\left(Tm^{-1}\right) \le \varepsilon',
    \]
    for a suitable $m=\poly(d)$ and $\varepsilon^{-1}=\poly(d)$. 
    A similar inequality holds for maximization variables $y_i$.
    The proof is concluded by normalizing all the coefficients to lie in $[-1,1]$. This will change the approximation error $\varepsilon$ by a factor polynomial in $d$.
\end{proof}

\section{Multi-linearization \& Interaction Degree Reduction}\label{sec:degree3interaction}

In this section, we reduce any degree-$2$ polynomial $f$ to a degree-$2$ multilinear polynomial in which each variable appears in at most three monomials. We extend the idea of the copy-gadget of \citet{hollender2025complexity}, by replacing each occurrence of a variable with a copy of it. In particular we prove that each copy $x_k^{(j)}$ of variable $x_k$, copies it's precedent copy $x_k^{(j-1)}$ (similarly for maximizing variables). The tricky part is to show recursively that the first copy $x_k^{(1)}$ collects all the gradients experienced by the other copies.

\begin{theorem}\label{th:degree2}
    \twoGDA is \PPAD-hard even for multilinear polynomials $f:[0,1]^d\times[0,1]^d\to[-1,1]$ in which each variable appears in at most three monomials and $\varepsilon^{-1}=\poly(d)$.
\end{theorem}

The remainder of the section is devoted to the proof of \Cref{th:degree2}. Consider a quadratic polynomial $f$ and recall that \twoGDA is \PPAD-Hard even for $\varepsilon=\frac{1}{\poly(d)}$. Moreover, recall that all the coefficients of $f$ are bounded by $1$ in absolute value.
Without loss of generality, we assume that $f$ has no degree-$0$ term (which would not affect the fixed points).  

\paragraph{Construction.} 
Our main idea is to replace each variable $u$ (which can be either $x_i$ or $y_i$) with a distinct variable, one for each occurrence.
Notice that each variable $u$ appears in at most $D:=2d+1$ monomials: $1$ in the linear term $u$, $2$ in the quadratic term $u^2$, and $2d-1$ with the other variables.
For convenience, index the $2d$ variables as $u_{1}, \ldots, u_{2d}$, where $u_k=x_k$ for $k \le d$ and $u_{k+d}=y_k$ for $k>d$.
Then, for each $u$ define the copies $u^{(1)},\ldots,u^{(D+1)}$ and the polynomial $f_{\textsf{copy}}$ as $f$ in which we replace each occurrence of $u$ with one of its copies $u^{(j)}$. The number of copies is $D+1$ since we need two copies for the quadratic terms $u_k$ and one copy for the remaining $D-1$ occurrences of $u$.

To do so, we define the operator $T_k$, which takes as input a polynomial and replaces every occurrence of the $u_k$ variable with one of its copies. Formally, given a function $g:[0,1]^{2d}\to[-1,1]$, which can be written as:\footnote{$u_{-k}$ is the set of variable excluding $u_k$.}
\[
g(u)= P_0(u_{-k})+u_k\sum_{\ell=1}^{D-1} P_\ell(u_{-k})+\alpha_ku_k^2,
\]
then 
\[
T_k(g)(u_{-k},u_k^{(1)},\ldots, u_k^{(D+1)})=P_0(u_{-k})+\sum_{\ell=1}^{D-1} u_k^{(\ell)} P_\ell(u_{-k})+\alpha_ku_k^{(D)}u_k^{(D+1)}.
\]
Finally, we can define $f_{\textsf{copy}}$ as $(T_{2d}\circ\cdots\circ T_2\circ T_1)(f)$. By construction $f_{\textsf{copy}}$ is a multilinear quadratic polynomial with $d'= 2d(D+1)=\poly(d)$ variables.

Now, it remains to link the copies of the same variable.
For each $k\le d$ and $j \in [D]$, we introduce a new variable $\tilde y^{(j)}_k$ which will interact with $x_k^{(j)}$ and $x_k^{(j+1)}$.  For each $k\le d$ and $j \in [D]$, we introduce a new variable $\tilde x^{(j)}_k$  which will interact with $y_k^{(j)}$ and $y_k^{(j+1)}$. Then, given two parameters $\mu=2\sqrt{\varepsilon'}$ and $\Lambda=\frac{5}{\sqrt{\varepsilon'}}$,  for any $k\le d$, we define the gadgets:
\[
 \cD_k(x,y)=\Lambda \sum_{j=1}^{D} \left(x^{(j+1)}_k-(1-2\mu)x^{(j)}_k-\mu\right)\cdot\left(\tilde y_k^{(j)}-\frac12\right),
\]
while for $d<k\le 2d$ we define the gadgets
\[
\cD_k(x,y)=\Lambda \sum_{j=1}^{D} \left(y^{(j+1)}_k-(1-2\mu)y^{(j)}_k-\mu\right)\cdot\left(\frac12-\tilde x_k^{(j)}\right).
\]
Finally, we can define
\[
\varphi(x,y)=\sum_{k=1}^{2d}\cD_k(x,y)
\]
and
\[
f'(x,y)=f_{\textsf{copy}}(x,y)+\varphi(x,y).
\]

Notice that, as promised, each newly introduced variable $u_k^i$ appears in exactly three monomials when introduced by the operator $T_k$, while the number of neighbors is not modified by the successive operators. The interaction graph of $f'$ is illustrated in \Cref{fig:figures}.
Finally, all coefficients are bounded by $\poly(d\Lambda)$. This, after a suitable normalization, guarantees that the constraints on the polynomial are satisfied.

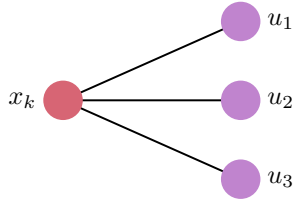
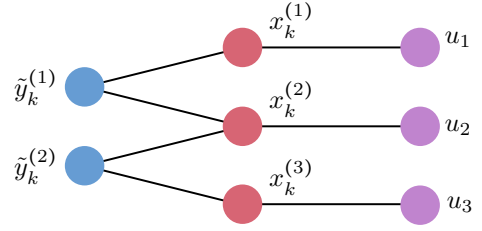
\begin{figure}
\centering
\begin{subfigure}[b]{0.5\textwidth}
    \centering
    \scalebox{0.99}{\tikzset{every picture/.style={line width=0.75pt}} %

\begin{tikzpicture}[x=0.75pt,y=0.75pt,yscale=-1,xscale=1]

\draw    (110,140) -- (200,140) ;
\draw    (110,140) -- (200,180) ;
\draw    (110,140) -- (200,100) ;
\draw  [draw opacity=0][fill={rgb, 255:red, 215; green, 101; blue, 116 }  ,fill opacity=1 ] (100,140) .. controls (100,134.48) and (104.48,130) .. (110,130) .. controls (115.52,130) and (120,134.48) .. (120,140) .. controls (120,145.52) and (115.52,150) .. (110,150) .. controls (104.48,150) and (100,145.52) .. (100,140) -- cycle ;
\draw  [draw opacity=0][fill={rgb, 255:red, 192; green, 132; blue, 206 }  ,fill opacity=1 ] (190,100) .. controls (190,94.48) and (194.48,90) .. (200,90) .. controls (205.52,90) and (210,94.48) .. (210,100) .. controls (210,105.52) and (205.52,110) .. (200,110) .. controls (194.48,110) and (190,105.52) .. (190,100) -- cycle ;
\draw  [draw opacity=0][fill={rgb, 255:red, 192; green, 132; blue, 206 }  ,fill opacity=1 ] (190,140) .. controls (190,134.48) and (194.48,130) .. (200,130) .. controls (205.52,130) and (210,134.48) .. (210,140) .. controls (210,145.52) and (205.52,150) .. (200,150) .. controls (194.48,150) and (190,145.52) .. (190,140) -- cycle ;
\draw  [draw opacity=0][fill={rgb, 255:red, 192; green, 132; blue, 206 }  ,fill opacity=1 ] (190,180) .. controls (190,174.48) and (194.48,170) .. (200,170) .. controls (205.52,170) and (210,174.48) .. (210,180) .. controls (210,185.52) and (205.52,190) .. (200,190) .. controls (194.48,190) and (190,185.52) .. (190,180) -- cycle ;

\draw (98,140) node [anchor=east] [inner sep=0.75pt]    {$x_{k}$};
\draw (212,100) node [anchor=west] [inner sep=0.75pt]    {$u_{1}$};
\draw (212,140) node [anchor=west] [inner sep=0.75pt]    {$u_{2}$};
\draw (212,180) node [anchor=west] [inner sep=0.75pt]    {$u_{3}$};

\end{tikzpicture}}
    \caption{Neighborhood of $x_k$ in the dependency graph of $f$.}
    \label{fig:first}
\end{subfigure}
\hfill
\begin{subfigure}[b]{0.45\textwidth}
    \centering
    \scalebox{0.99}{\tikzset{every picture/.style={line width=0.75pt}} %

\begin{tikzpicture}[x=0.75pt,y=0.75pt,yscale=-1,xscale=1]

\draw    (110,150) -- (190,130) ;
\draw    (110,150) -- (190,170) ;
\draw    (110,110) -- (190,90) ;
\draw    (110,110) -- (190,130) ;
\draw    (190,130) -- (280,130) ;
\draw    (190,169.5) -- (280,169.5) ;
\draw    (190,90) -- (280,90) ;
\draw  [draw opacity=0][fill={rgb, 255:red, 215; green, 101; blue, 116 }  ,fill opacity=1 ] (180,90) .. controls (180,84.48) and (184.48,80) .. (190,80) .. controls (195.52,80) and (200,84.48) .. (200,90) .. controls (200,95.52) and (195.52,100) .. (190,100) .. controls (184.48,100) and (180,95.52) .. (180,90) -- cycle ;
\draw  [draw opacity=0][fill={rgb, 255:red, 215; green, 101; blue, 116 }  ,fill opacity=1 ] (180,130) .. controls (180,124.48) and (184.48,120) .. (190,120) .. controls (195.52,120) and (200,124.48) .. (200,130) .. controls (200,135.52) and (195.52,140) .. (190,140) .. controls (184.48,140) and (180,135.52) .. (180,130) -- cycle ;
\draw  [draw opacity=0][fill={rgb, 255:red, 215; green, 101; blue, 116 }  ,fill opacity=1 ] (180,169.5) .. controls (180,163.98) and (184.48,159.5) .. (190,159.5) .. controls (195.52,159.5) and (200,163.98) .. (200,169.5) .. controls (200,175.02) and (195.52,179.5) .. (190,179.5) .. controls (184.48,179.5) and (180,175.02) .. (180,169.5) -- cycle ;
\draw  [draw opacity=0][fill={rgb, 255:red, 102; green, 156; blue, 211 }  ,fill opacity=1 ] (100,110) .. controls (100,104.48) and (104.48,100) .. (110,100) .. controls (115.52,100) and (120,104.48) .. (120,110) .. controls (120,115.52) and (115.52,120) .. (110,120) .. controls (104.48,120) and (100,115.52) .. (100,110) -- cycle ;
\draw  [draw opacity=0][fill={rgb, 255:red, 102; green, 156; blue, 211 }  ,fill opacity=1 ] (100,150) .. controls (100,144.48) and (104.48,140) .. (110,140) .. controls (115.52,140) and (120,144.48) .. (120,150) .. controls (120,155.52) and (115.52,160) .. (110,160) .. controls (104.48,160) and (100,155.52) .. (100,150) -- cycle ;
\draw  [draw opacity=0][fill={rgb, 255:red, 192; green, 132; blue, 206 }  ,fill opacity=1 ] (270,130) .. controls (270,124.48) and (274.48,120) .. (280,120) .. controls (285.52,120) and (290,124.48) .. (290,130) .. controls (290,135.52) and (285.52,140) .. (280,140) .. controls (274.48,140) and (270,135.52) .. (270,130) -- cycle ;
\draw  [draw opacity=0][fill={rgb, 255:red, 192; green, 132; blue, 206 }  ,fill opacity=1 ] (270,170) .. controls (270,164.48) and (274.48,160) .. (280,160) .. controls (285.52,160) and (290,164.48) .. (290,170) .. controls (290,175.52) and (285.52,180) .. (280,180) .. controls (274.48,180) and (270,175.52) .. (270,170) -- cycle ;
\draw  [draw opacity=0][fill={rgb, 255:red, 192; green, 132; blue, 206 }  ,fill opacity=1 ] (270,90) .. controls (270,84.48) and (274.48,80) .. (280,80) .. controls (285.52,80) and (290,84.48) .. (290,90) .. controls (290,95.52) and (285.52,100) .. (280,100) .. controls (274.48,100) and (270,95.52) .. (270,90) -- cycle ;

\draw (202,86.6) node [anchor=south west] [inner sep=0.75pt]    {$x_{k}^{( 1)}$};
\draw (291,86) node [anchor=west] [inner sep=0.75pt]    {$u_{1}$};
\draw (291,130.5) node [anchor=west] [inner sep=0.75pt]    {$u_{2}$};
\draw (292,169.5) node [anchor=west] [inner sep=0.75pt]    {$u_{3}$};
\draw (202,126.6) node [anchor=south west] [inner sep=0.75pt]    {$x_{k}^{( 2)}$};
\draw (202,166.1) node [anchor=south west] [inner sep=0.75pt]    {$x_{k}^{( 3)}$};
\draw (98,110) node [anchor=east] [inner sep=0.75pt]    {$\tilde{y}_{k}^{( 1)}$};
\draw (98,150) node [anchor=east] [inner sep=0.75pt]    {$\tilde{y}_{k}^{( 2)}$};

\end{tikzpicture}}
    \caption{Transformed dependency graph in $f'$.}
    \label{fig:third}
\end{subfigure}
        
\caption{Each node corresponds to a variable. A \textcolor[HTML]{D76574}{\textbf{red}} node represents a min variable, a \textcolor[HTML]{669CD3}{\textbf{blue}} node represents a max variable, while a \textcolor[HTML]{C084CE}{\textbf{purple}} node represents either a min or a max variable.}
\label{fig:figures}
\end{figure}

\paragraph{Completeness.} 
We show that a $\varepsilon'$-approximate solution to $f'$ (with a suitable $\varepsilon'=\frac{1}{\poly(d)}$) can be used to compute a $\varepsilon$-approximate solution to $f$. We first establish that the copies of each variable are indeed close to one another.

\begin{lemma}\label{lem:closness}
     For  every $\varepsilon'$-approximate solution to \twoGDA on $f'$ and all $j\in[D]$, it holds
    \[
    |x^{(j+1)}_k-x^{(j)}_k(1-2\mu)-\mu|\le \frac\mu2,\quad\text{and}\quad|y^{(j+1)}_k-y^{(j)}_k(1-2\mu)-\mu|\le\frac\mu2.
    \]
\end{lemma}
\begin{proof}
    We prove the bound $|x^{(j+1)}_k-x^{(j)}_k(1-2\mu)-\mu|\le \frac\mu2$; the other case is symmetric.
    Then, we prove that both $x^{(j+1)}_k-x^{(j)}_k(1-2\mu)-\mu> \frac\mu2$ and $x^{(j+1)}_k-x^{(j)}_k(1-2\mu)-\mu<-  \frac\mu2$ lead to a contradiction.
    
    \paragraph{Case 1.} Assume that $x^{(j+1)}_k-(1-2\mu)x^{(j)}_k-\mu> \frac\mu2$. The KKT condition for $\tilde y_k^{(j)}$ with respect to $1$ gives:
    \[
    \partial_{\tilde y_k^{(j)}} f'(x,y)(1-\tilde y_k^{(j)})=\Lambda\big(x^{(j+1)}_k-(1-2\mu)x^{(j)}_k-\mu\big)(1-\tilde y_k^{(j)})\le \varepsilon',
    \]
    and thus 
    \[
    1-\tilde y_k^{(j)}\le \frac{2\varepsilon'}{\Lambda\mu}\le \frac\mu2.%
    \]
    The partial derivative of $f'$ with respect to $x_k^{(j+1)}$ then satisfies
    \begin{align*}
    \partial_{x_k^{(j+1)}}f'(x,y)&=\partial_{x_k^{(j+1)}}f_{\textsf{copy}}(x,y)+\partial_{x_k^{(j+1)}}\varphi(x,y)\\
    &\ge -1+\Lambda \Big(\tilde y_k^{(j)}-\frac12\Big)-\Lambda(1-2\mu)\Big(\tilde y_k^{(j+1)}-\frac12\Big)\mathbb{I}(j\neq D)\tag{$|\partial_{x_k^{(j+1)}}f_{\textsf{copy}}|\le 1$}\\
    &\ge -1 +\Lambda\Big(\frac12-\frac12\mu\Big)-\Lambda\frac12(1-2\mu)\tag{$\tilde y_k^{(j+1)}\le 1$ and $1-\tilde y_k^{(j)}\le\frac\mu2$}\\
    &\ge -1+\frac12\Lambda\mu\ge \frac14\Lambda\mu.%
    \end{align*}
    Thus, thanks to the KKT condition of $x_k^{(j+1)}$ with respect to $0$ we obtain
    \(
    x_k^{(j+1)}\le \frac{4\varepsilon'}{\Lambda\mu}<\mu.%
    \)
    However, our assumption implies
    \[
    x_k^{(j+1)}>\frac32\mu+(1-2\mu)x_k^{(j)}\ge \frac32\mu,
    \]
    reaching a contradiction.

    \paragraph{Case 2.} Assume that $x^{(j+1)}_k-(1-2\mu)x^{(j)}_k-\mu< -\frac\mu2$. The KKT condition for $\tilde y_k^{(j)}$ with respect to $0$ guarantee that
    \[
    -\partial_{\tilde y_k^{(j)}} f'(x,y)\cdot \tilde y_k^{(j)}=-\Lambda\big(x^{(j+1)}_k-(1-2\mu)x^{(j)}_k-\mu\big)\cdot \tilde y_k^{(j)}\le \varepsilon',
    \]
    from which it follows  that 
    \[
    \tilde y_k^{(j)}\le \frac{\varepsilon'}{\mu\Lambda }\le \frac\mu2.%
    \]
    Now consider the derivative with respect to $x_k^{(j+1)}$ which is
    \begin{align*}
    \partial_{x_k^{(j+1)}}f'(x,y)&=\partial_{x_k^{(j+1)}}f_{\textsf{copy}}(x,y)+\partial_{x_k^{(j+1)}}\varphi(x,y)\\
    &\le1+\Lambda \Big(\tilde y_k^{(j)}-\frac12\Big)-\Lambda(1-2\mu)\Big(\tilde y_k^{(j+1)}-\frac12\Big)\mathbb{I}(j\neq d(x_k)-1)\tag{$|\partial_{x_k^{(j+1)}}f_{\textsf{copy}}|\le 1$}\\
    &\le 1-\frac{\mu\Lambda}{2}\le -\frac{\mu\Lambda}{4}.%
    \end{align*}
   The KKT of $x_k^{(j+1)}$ with respect to $1$ gives that
    \(
    x_k^{(j+1)}\ge1-\frac{4\varepsilon'}{\mu\Lambda}\ge 1-\mu,
    \)
    however, our assumption implies $x_k^{(j+1)}< 1-\frac32\mu$, reaching a contradiction.

    \paragraph{Conclusion.} Combining the two cases, we obtain that :
    \[|x^{(j+1)}_k-(1-2\mu)x^{(j)}_k-\mu|\le \frac\mu2,\]
    concluding the proof.
\end{proof}

To conclude the proof we show that 
\begin{align}\label{eq:hatx}
(\hat x,\hat y)=(x_1,\ldots, x_d,y_1,\ldots, y_d)=(x_1^{(1)},\ldots,x_d^{(1)},y_1^{(1)},\ldots,y_d^{(1)}).
\end{align}
represents a $\varepsilon$-approximate solution to the original \twoGDA instance.
\Cref{lem:closness} let us conclude that 
\begin{align}\label{eq:closetwoinstances}
|\hat x_k-x_k^{(j)}|\le 3D\mu \quad \text{and} \quad |\hat y_k-y_k^{(j)}|\le 3D\mu     \quad k\in[d], j \in [D+1].
\end{align}
Then, we prove that the derivative of $f'$ are close to those of $f$.

\begin{lemma}[Backpropagation]\label{claim:deriv}
Consider a solution $(x,y)$ to \twoGDA on $f'$ with approximation $\varepsilon'$ and let $(\hat x,\hat y)$ be defined according to \Cref{eq:hatx}. Then, $\forall k \in [d]$:
\begin{align*}
&\left|\partial_{x_k}f(\hat x,\hat y)-\partial_{x_k^{(1)}}f'(x,y)\right|\le 13\mu D^2+\frac{4D\varepsilon'}{\mu}\quad\text{and}\quad\left|\partial_{y_k}f(\hat x,\hat y)-\partial_{y_k^{(1)}}f'(x,y)\right|\le 13\mu D^2+\frac{4D\varepsilon'}{\mu}.
\end{align*}  
\end{lemma}
\begin{proof}
    We prove the statement for a variable $x_k$, $k \in [d]$; the proof for variables $y_k$ is symmetric.
    \paragraph{Bounding $\left|\partial_{x_k^{(1)}}f'(x,y)-\sum_{j=1}^{D+1}\partial_{x_k^{(j)}}f_{\textsf{copy}}(x,y)\right|$.} 
    We start observing that all copies beside the first one are bounded away from $0$ and $1$.
    Indeed, from \Cref{lem:closness}, we get 
    \[x_k^{(j)}\in[\tfrac\mu2,1-\tfrac\mu2] \quad \forall 2 \le j \le  D+1.\] 
    This holds since $(1-2\mu)\xi+\mu\in[\mu,1-\mu]$ for all $\xi\in[0,1]$.
   Then, the KKT condition implies:
    \[
    \left|\partial_{x_k^{(j)}}f'(x,y)\right|\le \frac{2\varepsilon'}\mu,\,\forall  2 \le j \le  D+1.
    \]
    Let $\alpha=(1-2\mu)$.
    Notice that
    \[
    \begin{cases}
    \partial_{x_k^{(1)}} \varphi(x,y)= -\Lambda \alpha\left(\tilde y_{k}^{(1)}-\frac{1}{2}\right)  \\
     \partial_{x_k^{(j)}} \varphi(x,y)=  \Lambda \left(\tilde y_{k}^{(j-1)}-\frac{1}{2}\right)  - \Lambda \alpha  \left(\tilde y_{k}^{(j)}-\frac{1}{2}\right)  \quad \textnormal{for } 2\le j \le D\\
    \partial_{x_k^{(D+1)}} \varphi(x,y)= \Lambda \left(\tilde y_{k}^{(D)}-\frac{1}{2}\right) 
     \end{cases}.
     \]
     Now, we prove by induction that for each $i \in [D]$ it holds:
    \[ \sum_{j\in [i]} \alpha^{j-1} \partial_{x_k^{(j)}} \varphi(x,y)= - \Lambda \alpha^i \left(\tilde y_{k}^{(i)}-\frac{1}{2}\right).\]
    This is clearly the case for $i=1$.
    Then, by induction, we get 
     \begin{align*}
     \sum_{j\in [i]} \alpha^{j-1} \partial_{x_k^{(j)}} \varphi(x,y)&= \sum_{j\in [i-1]}  \alpha^{j-1} \partial_{x_k^{(j)}} \varphi(x,y) + \alpha^{i-1} \left( \Lambda  \left(\tilde y_{k}^{(i-1)}-\frac{1}{2}\right)- \Lambda \alpha \left(\tilde y_{k}^{(i)}-\frac{1}{2}\right) \right)\\
     &= - \Lambda \alpha^{i-1} \left(\tilde y_{k}^{(i-1)}-\frac{1}{2}\right) + \alpha^{i-1} \left( \Lambda  \left(\tilde y_{k}^{(i-1)}-\frac{1}{2}\right)- \Lambda \alpha \left(\tilde y_{k}^{(i)}-\frac{1}{2}\right) \right)\\
     & = - \Lambda \alpha^{i} \left(\tilde y_{k}^{(i)}-\frac{1}{2}\right).
     \end{align*}
     Hence, we conclude that:
    \[ \sum_{j\in [D+1]} \alpha^{j-1} \partial_{x_k^{(j)}} \varphi(x,y) =   \sum_{j\in [D]} \alpha^{j-1} \partial_{x_k^{(j)}} \varphi(x,y) + \alpha^D   \partial_{x_k^{(j)}} \varphi(x,y)= 0\]
    Thus, we get:
    \begin{align*}
     \left|\partial_{x_k^{(1)}}f'(x,y)-\sum_{j=1}^{D+1}\partial_{x_k^{(j)}}f_{\textsf{copy}}(x,y)\right| & \le \left|\partial_{x_k^{(1)}}f'(x,y)-\sum_{j=1}^{D+1} \alpha^{j-1} \partial_{x_k^{(j)}}f_{\textsf{copy}}(x,y)\right| + \left| \sum_{j=1}^{D+1} (1-\alpha^{j-1})\partial_{x_k^{(j)}}f_{\textsf{copy}}(x,y) \right|  \\
     &\le  \left|\partial_{x_k^{(1)}}f'(x,y)-\sum_{j=1}^{D+1}\alpha^{j-1} \left( \partial_{x_k^{(j)}}f'(x,y)- \partial_{x_k^{(j)}} \varphi(x,y)\right)\right|+ 4\mu D^2\\ 
     & \le  \left|\sum_{j=1}^{D+1}\alpha^{j-1}  \partial_{x_k^{(j)}} \varphi(x,y)\right|+ 4\mu D^2+ \frac{4D\varepsilon'}{\mu}\\
     & =4\mu D^2+ \frac{4D\varepsilon'}{\mu},
    \end{align*}
    where in the second inequality we used $1-(1-2\mu)^{D}\le 2\mu D$.

    \paragraph{Bounding $\left|\partial_{\hat x_k} f(\hat x,\hat y)-\sum_{j=1}^{D+1}\partial_{x_k^{(j)}}f_{\textsf{copy}}(x,y)\right|$.} 
    Once we fix the index $k\in [d]$, we can write 
    \[
    \partial_{\hat x_k}f(\hat x,\hat y)= \sum_{j\in [D-1]:j\neq k } \alpha_j \hat u_j+ \beta +2\gamma \hat x_k,
    \] 
    while  
    \[
    \sum_{j=1}^{D+1} \partial_{x_k^{(j)}}f_{\textsf{copy}}(x,y)=\sum_{j\in[D-1]:j\neq k} \alpha_j u_j^{(i_j)}+\beta+\gamma (x_k^{(D)}+x_k^{(D+1)}),
    \]
    for some indexes $\{i_j\}_{j \neq k}$.
    Now, thanks to \Cref{eq:closetwoinstances}, we get:
    \begin{align*}
        \left|\partial_{\hat x_k}f(\hat x,\hat y)-\sum_{j=1}^{D+1}\partial_{x_k^{(j)}}f_{\textsf{copy}}(x,y)\right|&\le \sum_{j\neq k} |\alpha_j|\cdot |\hat u_j-u_j^{(i_j)}|+|\gamma|\cdot\left|2\hat x_k-\left(x_k^{(D)}+x_k^{(D+1)}\right)\right|\\
        &\le 3D^2\mu+6D\mu\le 9D^2\mu 
    \end{align*}
    \paragraph{Conclusion.} We can combine the previous two bounds to conclude that
    \begin{align*}
    \left|\partial_{x_k}f(\hat x,\hat y)-\partial_{x_k^{(1)}}f'(x,y)\right|&\le \left|\partial_{\hat x_k}f(\hat x,\hat y)-\sum_{j=1}^{D+1}\partial_{x_k^{(j)}}f_{\textsf{copy}}(x,y)\right|+\left|\sum_{j=1}^{D+1}\partial_{x_k^{(j)}}f_{\textsf{copy}}(x,y)-\partial_{x_k^{(1)}}f'(x,y)\right|\\
    &\le 13\mu D^2+\frac{4D\varepsilon'}{\mu}. 
    \end{align*}
    This concludes the proof of the lemma.
\end{proof}

Now, by \Cref{claim:deriv} and \Cref{eq:closetwoinstances}, we can easily verify that for all $k\le d$ and $x'\in[0,1]$ 
\begin{align*}
-\partial_{x_k}f(\hat x,\hat y)(x'-\hat x_k)&=-\partial_{x^{(1)}_k}f'(x,y)(x'-\hat x_k)+(\partial_{x^{(1)}_k}f'(x,y)-\partial_{x_k}f(\hat x,\hat y))(x'-\hat x_k)\\
&\le -\partial_{x^{(1)}_k}f'(x,y)(x'-\hat x_k)+13\mu D^2+\frac{4D\varepsilon'}{\mu}\\
&= -\partial_{x^{(1)}_k}f'(x,y)(x'- x_k^{(1)})+13\mu D^2+\frac{4D\varepsilon'}{\mu}\\
&\le \varepsilon'+13\mu D^2+\frac{4D\varepsilon'}{\mu}\\
&\le O(D^2 \sqrt{\varepsilon'} )\le \varepsilon,
\end{align*}
for $\varepsilon'$ polynomially small.
A similar result holds for the maximizer variables, showing that $(\hat x,\hat y)$ is a solution to the original problem.

\begin{remark}\label{rem:colorable}
    Note that each variable $x^{(j)}_k$ appears in at most $3$ monomials, and at most two of those ($\tilde y_k^{(j-1)}$ and $\tilde y_k^{(j)}$) are max variables, and the other one can either be a min or max variable (we also refer to \Cref{fig:figures}) for a visualization. The same holds for the max variables. Thus, every max (resp.~min) variable is connected with at most one other max (resp.~min) variable. This will be instrumental in \Cref{sec:2vs2}.
\end{remark}

\section{From Multilinear Degree-$2$ to Polymatrix}\label{sec:polymatrix}

In \Cref{sec:degree3interaction} we showed hardness for \twoGDA for a degree-$2$ multilinear polynomial in which each variable appears in at most three monomials.
To reduce this to \TTPG, we adapt the reduction of \citet{hollender2025complexity}. The original reduction distributes degree-$1$ terms across all variables of the opposing team, which would inflate the interaction degree and would not yield a polymatrix instance in which the interaction graph has degree at most three.
We fix this issue by a simple modification: we assign each linear term to an already-existing edge, so that the interaction graph of the polynomial coincides exactly with the graph of the polymatrix instance.
We also assume, without loss of generality, that there are no isolated nodes, i.e., every variable appears in at least one degree-$2$ monomial.

\begin{theorem}\label{th:manyplayerpolymatrix}
\TTPG is \PPAD-hard even when each player interacts with at most $3$ other players, each player has $2$ actions, and $\varepsilon^{-1}=\poly(n)$.
\end{theorem}

\begin{proof}
Let $f:[0,1]^d\times[0,1]^d\to[-1,1]$ be a multilinear degree-$2$ polynomial in which each variable appears in at most three monomials. We introduce a player of the $X$ team for each of the $d$ minimizing variables and a player for the $Y$ team for each of the $d$ maximizing variables. The game thus has $n = 2d$ players. Then, we can write the original multilinear polynomial as
\[
f(z)=\sum_{v\in V} \ell_v z_v+\sum_{\{u,v\}\in E} c_{\{u,v\}} z_{u}z_{v},
\]

where $V=X\sqcup Y$. It is also convenient to define the team sign as
\[
\tau_v=\begin{cases}
    -1&\text{if }v\in X\\
    +1&\text{if }v\in Y.
\end{cases}
\]

To distribute linear terms across edges, for each $v\in V$, fix any incident edge $h(v)\in E$, and for each edge $e=\{u,v\}\in E$ set
\[\beta_{e,u}=
\begin{cases}
    \ell_u&\text{if }h(u)=e\\
    0,&\text{otherwise}
\end{cases},\quad\text{and}\quad
\beta_{e,v}=
\begin{cases}
    \ell_v&\text{if }h(v)=e\\
    0,&\text{otherwise}
\end{cases},
\]
and $\gamma_e=c_{\{u,v\}}$. Then, the payoff matrix for each $e=\{u,v\}\in E$ is defined via:
\[
A^{u,v}=\tau_u B_{u,v},\quad\text{and}\quad A^{v,u}=\tau_v B_{u,v}^\top.
\]
where 
\[
B_{u,v}=\begin{pmatrix}
    0&\beta_{e,v}\\
    \beta_{e,u}&\gamma_e+\beta_{e,u}+\beta_{e,v}
\end{pmatrix}.
\]

Clearly, we can normalize $B_{u,v}$ so that all entries lie in $[-1,1]$. Note that this might require us to decrease $\varepsilon$ by a polynomial factor, but this does not alter the result.

We can easily see that this is indeed a \TTPG instance in which each player interacts with at most three other players. Indeed, if $u,v\in X$ then $\tau_v=\tau_u=-1$ and $A^{u,v} = (A^{v,u})^\top$ (and similarly if $u,v\in Y$), while if $u\in X$ and $v\in Y$, then $\tau_u=-1=-\tau_v$ which implies that $A^{u,v} = -(A^{v,u})^\top$ (and similarly if $u\in Y$ and $v\in X$).

It remains to show that from every equilibrium of the \TTPG instance we can recover a solution of the original \twoGDA instance.
 For each player $v$, let $z_v$ denote the probability they assign to their second action.
For an edge $e=\{u,v\}$, we can define the quantity
\[
\Phi(e) = (1-z_u)z_v\beta_{e,v}+z_u(1-z_v)\beta_{e,u}+z_uz_v(\gamma_e+\beta_{e,u}+\beta_{e,v})=\gamma_ez_uz_v+\beta_{e,v}z_v+\beta_{e,u}z_u,
\]
which is the expected payoff contribution of that edge (modulo its sign).
Summing over all edges and using the definition of $\beta_e$, we get
\begin{align*}
    \sum_{e\in E} \Phi(e)&=\sum_{\{u,v\}\in E}c_{u,v}z_uz_v+\sum_{\{u,v\}\in E}(\beta_{e,v}z_v+\beta_{e,u}z_u)\\
    &=\sum_{\{u,v\}\in E}c_{u,v}z_uz_v+\sum_{u\in V}z_u\sum_{v:\{u,v\}\in E} \beta_{e,u}+\sum_{v\in V}z_v\sum_{u:\{u,v\}\in E} \beta_{e,v}\\
    &=\sum_{\{u,v\}\in E}c_{u,v}z_uz_v+\sum_{u\in V} z_u\ell_u\\
    &=f(z).
\end{align*}

The proof is now almost concluded. 
The equilibrium condition of a player $v$ with respect to $z'_v$ implies:
\[
U_v(z_v',z_{-v})-U_v(z)=\tau_v(f(z_v',z_{-v})-f(z))\le \varepsilon.
\]
Moreover, since $f$ is multilinear:
\[
\tau_v(f(z_v',z_{-v})-f(z))=\tau_v\partial_{z_v}f(z)(z_v'-z_v)\le\varepsilon,
\]
which shows that we have found a solution to the original \twoGDA instance.
\end{proof}

\section{Reduction to $2$ vs.~$2$ \TTPG}\label{sec:2vs2}

In \Cref{sec:polymatrix}, we reduced the problem to a many-player two-team zero-sum polymatrix game. As we observed in \Cref{rem:colorable}, each variable appears in at most one monomial with another variable of the same team. Since the reduction of \Cref{sec:polymatrix} preserves the interaction graph, the two teams $X$ and $Y$ of the \TTPG instance constructed there are $2$-colorable. Thus, we can partition $X=X_1\sqcup X_2$ and $Y=Y_1\sqcup Y_2$ such that for all $u,v\in X_i$ or $u,v\in Y_i$, then $u$ and $v$ are not connected. Following the standard reduction from polymatrix games to bi-matrix games \citep{althofer1994sparse,babichenko2016can,rubinstein2017settling}, we introduce one player $P$ for each group $\{X_1, X_2, Y_1, Y_2\}$. A strategy for the player in the new game corresponds to selecting a player from the original polymatrix game, one of the two actions assigned to it, and one vertex of the opposing team. Adding the hide-and-seek gadget of \citet{althofer1994sparse} guarantees that each player $P$ randomizes uniformly over their choice of original player. The proof is standard, but we include an overview for completeness.

\begin{theorem}\label{th:2vs2}
    \TTPG is \PPAD-hard even when there are $2$ players in each team and $\varepsilon^{-1}= \poly(n)$.
\end{theorem}

\begin{proof}
Let $N_P=|P|$ for every $P\in\{X_1,X_2,Y_1,Y_2\}:=\mathcal{N}$. For each player $P \in \mathcal{N}$, we define its complement $\bar P$ as the player on the opposing team in the corresponding position (so $\bar X_1=Y_1$ and $\bar Y_1=X_1$). A pure action for a player $P$ is a tuple $(u,b,v)\in[N_P]\times\{0,1\}\times[N_{\bar P}]$, to be interpreted as: player $P$ simulates original player $u$ playing action $b\in\{0,1\}$, while trying to ``catch'' the hiding spot $v$ chosen by player $\bar P$.

Let $z=(z_P)_{P}$ be a mixed strategy for the $4$-player games. For each $P$ and $u\in P$ let $M_P(u;z)$ denote the marginal probability under $z_P$ of selecting an action $(u,\cdot,\cdot)$, and $p_P(u;z)$ denote the marginal probability of selecting $(u,1,\cdot)$. Formally:
\[
M_P(u;z)=\sum_{b\in\{0,1\}, v\in \bar P} z_P(u,b,v),
\quad
\text{and} 
\quad
p_P(u;z)=\sum_{ v\in \bar P} z_P(u,1,v).
\]
Moreover, we let $P(u)$ be the group in $\mathcal{N}$ that contains $u$.

Setting $\tau_P=-1$ for $P\in\{X_1,X_2\}$ and $\tau_P=1$ for $P\in\{Y_1,Y_2\}$, the payoff of a pure action profile  $a=(u_P,b_P,v_P)_{P\in\mathcal{N}}\in \bigtimes_{P\in\mathcal{N}}P\times\{0,1\}\times\bar P$ for the $P$ player is 
\[
U_P(a) = \tau_P(\Psi_{\textsf{sim}}(a)+\Psi_{\textsf{h\&s}}(a)),
\]
where, letting $\Phi_{\{i,j\}}(b,b')$ denote the payoff of the original binary action game on edge $\{i,j\}$ with $b,b'\in\{0,1\}$,
\[
\Psi_{\textsf{sim}}(a)= \sum_{\{i,j\}\in E} N_{P(i)}N_{P(j)}\mathbb{I}(u_{P(i)}=i,u_{P(j)}=j)\Phi_{\{i,j\}}(b_{P(i)},b_{P(j)}),
\]
and the hide-and-seek component is
\[
\Psi_{\textsf{h\&s}}(a) = B\Bigg(\sum_{P\in\{X_1,X_2\}}\mathbb{I}(u_{P}=v_{\bar P})-\sum_{P\in\{Y_1,Y_2\}}\mathbb{I}(u_{P}=v_{\bar P})\Bigg)
\]
for a large $B=\poly(|V|,1/\varepsilon,1/\delta)$.

Let $z$ be an $\varepsilon'$-equilibrium of the $4$ player game. We claim that setting 
\[
\hat z_u = p_P(u;z)/M_P(u;z),
\]
for all $u\in V$ (with $\hat z_u$ defined arbitrarily when $M_P(u;z)=0$), we obtain a solution to the original \TTPG instance. 
By standard calculations \citep{althofer1994sparse,babichenko2016can,rubinstein2017settling}, for some $B=\poly(|V|,1/\varepsilon,1/\delta)$, the hide-and.seek gadget enforces 
\begin{align}\label{eq:uniform}
\left|M_P(u;z)-\frac{1}{N_P}\right|\le \delta,
\end{align}
Notice that choosing $\delta\le 1/(2|V|)$ ensures that that $M_P(u;z)\neq 0$ for all $u\in P$.

Now we show the correctness of the reduction. Fix any node $v\in P$ and consider any deviation of player $P$ to a mixed strategy $z'=(z'_P,z_{-P})$ that preserves the marginals $M_P(\cdot ; z)$ and the marginal over $\bar P$, i.e., $\sum_{u \in P, b \in \{0,1\}} z_P(u,b,v)$ for all $v \in \bar P$, but changes the conditional probability of playing action $1$ for player $u$ to $s\in[0,1]$. Formally, we consider a 
\[
z'_P(r,b,v)=\begin{cases}
    z_P(r,b,v)&\text{if }r \neq u\\
    (1-s)(z_P(r,0,v)+z_P(r,1,v))&\text{if }r = u, b=0\\
    s(z_P(r,0,v)+z_P(r,1,v))&\text{if }r = u,b=1
\end{cases}.
\]

Under this deviation, all $\hat z_w$ remain the same, except for player $u$, for which $\hat z'_u=s$. Indeed,
\[
p_P(u;z')=\sum_{v\in \bar P} z'_P(u,1,v)=s\cdot M_P(u,z),
\]
and $M_P(w;z')=M_P(w;z)$, for every $w$. Hence
\[
\mathsf{E}_{a\sim z'}[\Psi_{\textsf{h\&s}}(a)]=\mathsf{E}_{a\sim z}[\Psi_{\textsf{h\&s}}(a)].
\] 
Now observe that
\begin{align*}
    U_P(z')-U_P(z)&=\tau_P\left(
\mathsf{E}_{a\sim z'}[\Psi_{\textsf{sim}}(a)]-\mathsf{E}_{a\sim z}[\Psi_{\textsf{sim}}(a)]
\right)\\
&=\tau_PN_PM_P(u;z)\sum_{j:\{u,j\}\in E}N_{P(j)}M_{P(j)}(j;z)\cdot\left(\Phi_{\{u,j\}}(s,\hat z_j)-\Phi_{\{u,j\}}(\hat z_u,\hat z_j)\right)
\end{align*}

Applying \Cref{eq:uniform}, we have $N_QM_Q(w;z) = 1\pm O(\delta |V|)$ for all player $Q$ and every $w\in Q$.  Since the maximum degree of the original games is $3$, this gives
\[
U_P(z')-U_P(z)= \tau_P\sum_{j:\{u,j\}\in E}\left(\Phi_{\{u,j\}}(s,\hat z_j)-\Phi_{\{u,j\}}(\hat z_u,\hat z_j)\right)\pm O(\delta |V|).
\]
Since $z$ is $\varepsilon'$-approximate equilibrium of the $4$-player game, we have
\[
U_P(z')\le U_P(z)+\varepsilon'.
\]
Therefore, for every $u\in P$ and every $s\in[0,1]$,
\[
\tau_P\sum_{j:\{u,j\}\in E}\left(\Phi_{\{u,j\}}(s,\hat z_j)-\Phi_{\{u,j\}}(\hat z_u,\hat z_j)\right)\le \varepsilon'+O(\delta|V|).
\]
The left-hand side is precisely the unilateral gain of the original player $u$ when deviating from $z_u$ to $s$. Choosing $\varepsilon'$ and $\delta$ such that $\varepsilon'+O(\delta|V|)\le\varepsilon$, the profile $z$ is an $\varepsilon$-approximate equilibrium of the original \TTPG instance.
\end{proof}

\newpage
\appendix
\section{Missing Proofs from \Cref{sec:three}}\label{sec:missing3}

\xminusy*
\begin{proof}
    Consider the optimality conditions of variables $x_{i,j}^q$ with respect to $y_{i,j}^q$ and vice-versa. Namely
    \begin{subequations}
    \begin{align}
        -&\partial_{x_{i,j}^q}f(x,y)(y_{i,j}^q-x_{i,j}^q)\le\varepsilon\label{eq:tmp1}\\
         &\partial_{y_{i,j}^q}f(x,y)(x_{i,j}^q-y_{i,j}^q)\le\varepsilon\label{eq:tmp2}.
    \end{align}
    \end{subequations}
    Now, we define $A:=M_i-\bar x_q$, $\xi:=|x_{i,j}^q-y_{i,j}^q|$, and we substitute the expressions for $\partial_{x_{i,j}^q}f(x,y)$ and $\partial_{y_{i,j}^q}f(x,y)$ into \Cref{eq:tmp1,eq:tmp2}. Consider the case in which $A>0$, rearranging \Cref{eq:tmp1} we obtain:
    \begin{align*}
    2 A \xi^2&\le\varepsilon+\delta^2 s_q(x)\partial_{x_{i,j}^q} H_q(x,y)(y_{i,j}^q-x_{i,j}^q)\\
    &\le \varepsilon+\delta^2 s_q(x) |\partial_{x_{i,j}^q} H_q(x,y)|\cdot |y_{i,j}^q-x_{i,j}^q|\\
    &\le \varepsilon+3\delta^2 s_q(x)\xi m \tag{$|\partial_{x_{i,j}^q} H_q(x,y)|\le 3m$}
    \end{align*}
    and similarly for \Cref{eq:tmp2}, when $A<0$, we obtain:
    \begin{align*}
    -2A\xi^2\le \varepsilon+3\delta^2s_q(x)m\xi.
    \end{align*}
    In both cases it holds that $2|A|\xi^2\le \varepsilon+3\delta^2s_q(x)m\xi$ which can be solved in $\xi$ to obtain
    \[
    \xi\le\frac{3\delta^2s_q(x)m+\sqrt{(3\delta^2s_q(x)m)^2+8|A|\varepsilon}}{4|A|}\le \frac{3\delta^2s_q(x)m}{2|A|}+\sqrt{\frac\varepsilon{|A|}}\le \frac{2\delta^2s_q(x)m}{|A|}+\sqrt{\frac\varepsilon{|A|}},
    \]
    concluding the proof.
\end{proof}

\harmonic*
\begin{proof}
    On each side of $nx$, ordering the relevant integers by increasing distance, the
    $k$-th such integer has distance at least $k$. Formally:
    \begin{align*}
        \sum_{i\in[n]\setminus B(x)}\frac1{|M_i-x|}&=n\sum_{i\in[n]:|i-nx|\ge 1}\frac1{|i-nx|}\\
        &\le2n\sum_{k=1}^n\frac1k\le4n\log n,
    \end{align*}
    concluding the proof.
\end{proof}

\section{Missing Proofs from \Cref{sec:3to2}}\label{app:3to2}

\lemmaspecial*
\begin{proof}
    To do so, we need to take a generic cubic polynomial and replace the cubic terms with symmetric terms $L_k(x,y)^3$.
    
    Let $c\in\Reals$ and $u_1,u_2,u_3\in [0,1]$ be some variables (not necessarily different). Then 
    if $c\ge 0$, thanks to \Cref{lem:polarization}, we get
    \begin{align*}
        c\cdot u_1u_2u_3 &= \frac{9c}2\left(\frac{u_1+u_2+u_3}3\right)^3+\frac{4c}{3}\left[\left(1-\frac{u_1+u_2}{2}\right)^3+\left(1-\frac{u_2+u_3}{2}\right)^3+\left(1-\frac{u_1+u_3}{2}\right)^3\right]\\
        &-4c\left[\left(\frac{u_1+u_2}{2}\right)^2+\left(\frac{u_2+u_3}{2}\right)^2+\left(\frac{u_1+u_3}{2}\right)^2\right]\\
        &+4c(u_1+u_2+u_3-1)+\frac{c}{6}\left(u_1^3+u_2^3+u_3^3\right)
    \end{align*}

    If $c<0$, then, again thanks to \Cref{lem:polarization} and the algebraic identity that $-a^3=(1-a)^3-1+3a-3a^2$, we get:
    \begin{align*}
        c\cdot u_1u_2u_3 %
        &=\frac{4|c|}{3}\left[\left(\frac{u_1+u_2}{2}\right)^3+\left(\frac{u_2+u_3}{2}\right)^3+\left(\frac{u_1+u_3}{2}\right)^3\right]+\frac{9|c|}{2}\left[\left(1-\frac{u_1+u_2+u_3}{3}\right)^3-3\left(\frac{u_1+u_2+u_3}{3}\right)^2\right]\\
        &+\frac{|c|}{6}[(1-u_1)^3+(1-u_2)^3+(1-u_3)^3]-\frac{|c|}{2}(u_1^2+u_2^2+u_3^2)\\
        &+5|c|(u_1+u_2+u_3-1)
    \end{align*}

   The proof is concluded by observing that we can group the terms and normalize all the coefficients by decreasing $\varepsilon$ by a factor polynomial in $d$.
\end{proof}

\printbibliography

\end{document}